\documentclass[pdflatex,sn-mathphys-num]{sn-jnl}

\usepackage{graphicx}
\usepackage{multirow}
\usepackage{amsmath,amssymb,amsfonts}
\usepackage{amsthm}
\usepackage{mathrsfs}
\usepackage[title]{appendix}
\usepackage{xcolor}
\usepackage{textcomp}
\usepackage{manyfoot}
\usepackage{booktabs}
\usepackage{algorithm}
\usepackage{algorithmicx}
\usepackage{algpseudocode}
\usepackage{listings}
\usepackage{tikz}
\usetikzlibrary{quantikz}
\usepackage{afterpage}
\usepackage{wasysym}
\usepackage{booktabs}
\usepackage{arydshln}

\theoremstyle{thmstyleone}
\newtheorem{theorem}{Theorem}
\newtheorem{lemma}[theorem]{Lemma}

\theoremstyle{thmstylethree}
\newtheorem{assumption}{Assumption}
\newtheorem{problem}{Problem}
\newtheorem{definition}{Definition}

\newcommand{\proc}[1]{\textup{\textsf{#1}}}

\begin{document}

\title[Article Title]{Quantum algorithm for solving McKean-Vlasov stochastic differential equations}

\author*[1]{\fnm{Koichi} \sur{Miyamoto}}\email{miyamoto.kouichi.qiqb@osaka-u.ac.jp}

\affil*[1]{\orgdiv{Center for Quantum Information and Quantum Biology}, \orgname{The University of Osaka}, \orgaddress{\city{Toyonaka}, \state{Osaka}, \country{Japan}}}

\abstract{Quantum Monte Carlo integration (QMCI), a quantum algorithm for calculating expectations that provides a quadratic speed-up compared to its classical counterpart, is now attracting increasing interest in the context of its industrial and scientific applications.
In this paper, we propose the first application of QMCI to solving McKean-Vlasov stochastic differential equations (MVSDEs), a nonlinear class of SDEs whose drift and diffusion coefficients depend on the law $\mu_t$ of the solution $X_t$---appearing in fields such as finance and fluid mechanics.
We focus on the problem setting where the coefficients depend on $\mu_t$ through expectations of some functions $\mathbb{E}[\varphi_k(X_t)]$, and the goal is to compute the expectation of a function $\mathbb{E}[\phi(X_T)]$ at a terminal time $T$.
We devise a quantum algorithm that leverages QMCI to compute these expectations, combined with a high-order time discretization method for SDEs and extrapolation of the expectations in time.
The proposed algorithm estimates $\mathbb{E}[\phi(X_T)]$ with accuracy $\epsilon$, making $O(1/\epsilon^{1+2/p})$ queries to the quantum circuit for time evolution over one step, where $p\in(1,2]$ is the weak order of the SDE discretization method.
This demonstrates the speed-up over the well-known classical algorithm called the particle method with complexity of $O(1/\epsilon^3)$.
We conduct numerical demonstrations of our quantum algorithm applied to examples of MVSDEs, with some parts emulated classically, and observe that the accuracy and complexity behave as expected.}


\maketitle

\section{Introduction}
Applications of quantum algorithms to practical problems have been explored extensively.
One of such algorithms is quantum Monte Carlo integration (QMCI) \cite{montanaro2015} for estimating expectations of random variables.
It provides the quadratic speed-up compared to the classical counterpart: in estimation with accuracy $\epsilon$, QMCI makes $O(1/\epsilon)$ queries to a quantum circuit to encode the distribution of the random variable into a quantum state, while the classical Monte Carlo method has the sample complexity of $O(1/\epsilon^2)$. 
Applications of QMCI are widely explored in various fields such as derivative pricing in finance \cite{Rebentrost2018,Stamatopoulos2020optionpricingusing,Chakrabarti2021thresholdquantum,miyamoto2022bermudan,kaneko2022quantum,doriguello2022} and machine learning \cite{pmlr-v139-wang21w,wiedemann2023quantum,Wan_Zhang_Li_Zhang_Sun_2023,NEURIPS2023_401aa72e,hikima2024quantum}.
Following this trend, it is desirable to extend the scope of quantum applications to more general and advanced problem settings than those considered so far, thereby enhancing the benefits of quantum computing.

In this paper, we consider applying QMCI to McKean-Vlasov stochastic differential equations (MVSDE) \cite{McKean1966}, an extended type of stochastic differential equations (SDEs).
SDEs are widely used to describe the dynamics of random processes such as asset prices in derivative pricing.
With a suitably defined probability space, an SDE for an $\mathbb{R}^d$-valued random process $X_t$ is written as
\begin{align}
    dX_t=a(t,X_t) dt + b(t,X_t) dW_t,
\end{align}
where $W_t$ is a $m$-dimensional Brownian motion, and $a$ and $b$ are $\mathbb{R}^d$- and $\mathbb{R}^{d \times m}$-valued coefficient functions, respectively, and a typical goal is calculating an expectation $\mathbb{E}[\phi(X_T)]$ with some function $\phi:\mathbb{R}^d\rightarrow\mathbb{R}$ at a terminal time $T$.
MVSDEs are an extension of this in the following form:
\begin{align}
    dX_t=a_{\rm MV}(t,X_t,\mu_t) dt + b_{\rm MV}(t,X_t,\mu_t) dW_t.
\label{eq:MVSDEGen}
\end{align}
Now, the coefficient functions depend on $\mu_t$, the law of $X_t$.
In other words, the equation governing the solution depends on the solution itself, which makes MVSDEs nonlinear and addressing them more difficult than ordinary SDEs.
Nevertheless, MVSDEs are widely used to describe complicated dynamics of systems in various fields.
For example, derivative pricing under the stochastic local volatility model involving model calibration falls into an MVSDE~\cite{guyon2013nonlinear}.
Although technological advances such as GPUs have accelerated derivative pricing, further improving the accuracy and speed of pricing under advanced models, such as the SLV model, remains a long-lasting challenge in financial engineering.
MVSDEs also appear in fluid mechanics \cite{Bossy1997}, interacting particle systems \cite{Meleard1996}, and so on.

Although SDEs have been considered in the context of applications of QMCI, especially in derivative pricing \cite{Rebentrost2018,Stamatopoulos2020optionpricingusing,Chakrabarti2021thresholdquantum,kaneko2022quantum}, applying QMCI to MVSDEs has not been studied yet as far as the author knows.
As the first attempt in this direction, this paper considers a relatively simple setting, where the coefficient functions depend on $\mu_t$ through expectations of given functions, $\gamma_k(t)\coloneqq\mathbb{E}[\varphi_k(X_t)]$ \cite{Belomestny2018} (see Sec.~\ref{sec:prob} for the concrete form of the MVSDE), and the goal is calculating some expectation $\mathbb{E}[\phi(X_T)]$.
Although the problems to which we ultimately aim to apply quantum computing are more complex---including derivative pricing under the SLV model---we believe it is reasonable to begin with the above setting, given the complications specific to each advanced problem, which need to be addressed separately: for example, calibrating the SLV model involves the estimation of not the expectations as in this work, but the conditional expectations of the squared volatility, which are more demanding. 
In fact, even in the setting considered in this work, applying QMCI to MVSDEs is not straightforward, as seen later.
Thus, devising an end-to-end quantum algorithm for the setting and demonstrating the computational complexity reduction achieved by it is an important step toward the ultimate goal. 

Before the quantum algorithm, let us begin with reviewing classical algorithms.
A popular one is the particle method \cite{Bossy1997,bossy1996convergence,antonelli2002rate}.
Although we leave its detailed description to Sec.~\ref{sec:particle}, roughly speaking, it is a Monte Carlo-based step-by-step method: we generate many paths (particles) of $X_t$ up to a time grid point $t_i$, estimate $\gamma_k(t_i)$ as the average over the generated particles, and use it for time evolution to the next time point $t_{i+1}$.
According to \cite{antonelli2002rate}, to suppress the error to $\epsilon$, we generate $O(1/\epsilon^2)$ particles, as is common in Monte Carlo methods, with a time step size proportional to $\epsilon$.
This results in the computational complexity of $O(1/\epsilon^3)$, with time evolution over one interval regarded as a unit operation.

It looks straightforward to quantize this method as the following loop: once $\gamma_k(t)$ is estimated up to $t=t_i$, following \cite{Rebentrost2018,Stamatopoulos2020optionpricingusing,Chakrabarti2021thresholdquantum,kaneko2022quantum}, we can construct a quantum circuit for time evolution of $X_t$ to $t_{i+1}$, which generates a quantum state that encodes the distribution of $X_{t_{i+1}}$, and use it to estimate $\gamma_k(t)$ at $t=t_{i+1}$ by QMCI.
This seems to immediately provide a quadratic quantum speed-up, replacing the classical sample complexity of $O(1/\epsilon^2)$ with the quantum query complexity of $O(1/\epsilon)$.
However, things do not go so smoothly, due to an issue unique to quantum computing.
The distribution-encoded quantum state is destroyed by QMCI, and to estimate $\gamma_k(t)$ at each time point $t=t_i$, we need time evolution from the initial time.
This results in the complexity scaling as $O(n_{\rm t}^2)$ on the number $n_{\rm t}$ of time steps and as $O(1/\epsilon^2)$ on $\epsilon$ if the time step size is proportional to $\epsilon$, where the time evolution over one interval is regarded as a unit operation also in the quantum case.
This scaling on $n_{\rm t}$ is worse than $O(n_{\rm t})$ of the classical method, in which we can store the values of $X_t$ at intermediate times in memory and we do not need to restart time evolution from the beginning each time.
This issue has already been pointed out in previous papers \cite{doriguello2022,miyamoto2025}.
Consequently, the total complexity of the above quantum approach is $O(1/\epsilon^3)$, which shows no quantum speed-up.

To break through the situation, we reduce the number of time steps using a higher-order time discretization method for SDEs such as the stochastic Runge-Kutta (SRK) method \cite{Rossler01032006,Rossler2010} and a higher-order (concretely, linear) extrapolation of $\gamma_k(t)$.
This allows us to set the time step size to $O(\epsilon^{1/p})$, where $p\in(1,2]$.
Incorporating this into the above quantum approach yields the main proposal in this study, our quantum algorithm for solving MVSDEs.
We conduct a rigorous evaluation of our algorithm's accuracy and complexity, and show that it indeed outputs an expectation of $\mathbb{E}[\phi(X_T)]$ with accuracy $\epsilon$ at a cost of $O(1/\epsilon^{1+2/p})$ query complexity, which is an improved scaling compared to $O(1/\epsilon^3)$ of the classical method.
We also conduct a numerical experiment to demonstrate the performance of the proposed algorithm.
Applying our quantum algorithm to an example of MVSDEs with some parts emulated classically, we numerically confirm its accuracy and complexity scaling.

The rest of this paper is organized as follows.
As a preliminary part, Sec.~\ref{sec:prem} presents the problem considered in this paper and gives some reviews on MVSDEs and QMCI.
Sec.~\ref{sec:main} is the main part of this paper: after giving an intuitive explanation on how our algorithm works, we present its entire procedure and the main theorem on its accuracy and complexity, leaving the proof to the appendix.
Sec.~\ref{sec:demo} is devoted to the numerical experiment, and Sec.~\ref{sec:sum} concludes this paper.

\section{Preliminary \label{sec:prem}}

\subsection{Notation}

For $n\in\mathbb{N}$, we denote $[n]\coloneqq \{1,\ldots,n\}$.
$\mathbb{N}_0\coloneqq \{0\} \cup \mathbb{N}$ denotes the set of all nonnegative integers.
$\mathbb{R}_+$ denotes the set of all positive real numbers.

For a vector $x=(x_1,\ldots,x_d)\in\mathbb{C}^d$, $|x|\coloneqq \sqrt{\sum_{i=1}^d |x_i|^2}$ denotes its Euclidean norm, and for a matrix $A=(a_{ij})\in\mathbb{C}^{m \times n}$, $|A|\coloneqq \sqrt{\sum_{i,j=1}^d |a_{i,j}|^2}$ denotes its Frobenius norm.

$N(\mu,v)$ denotes the normal distribution with mean $\mu$ and variance $v$.

If $x\in\mathbb{R}$ satisfies $|x-y|\le\epsilon$ for $y\in\mathbb{R}$ and $\epsilon\in\mathbb{R}_+$, we say that $x$ is an $\epsilon$-approximation of $y$.

$\mathbb{C}^{n}_P(\mathbb{R}^d,\mathbb{R})$ denotes the space of $n$ times continuously differentiable functions from $\mathbb{R}^d$ to $\mathbb{R}$ for which all partial derivatives up to order $n$ have polynomial growth \cite{Rossler01032006}.
In this notation, we omit the domain $\mathbb{R}^d$ and codomain $\mathbb{R}$ when they are clear from context.

\subsection{Problem \label{sec:prob}}

In this paper, we address the following problem.

\begin{problem}
Let $[0,T]$ be a finite time interval and $(\Omega,\mathcal{F},P)$ be a complete probability space, where a standard $m$-dimensional Brownian motion $W_t=(W^1_t,\ldots,W^d_t)$ is defined.
Consider the $\mathbb{R}^d$-valued random process $X_t$ that has a deterministic initial value $X_0$ and obeys the Ito SDE 
    \begin{align}
    & dX_t= a(t,X_t) dt + b(t,X_t) dW_t \nonumber \\
 & a(t,X_t)=\sum_{k=1}^K 
 \gamma_k(t) \alpha_k(X_t),~b(t,X_t)=\sum_{k=1}^K 
 \gamma_k(t) \beta_k(X_t)
\label{eq:MVSDE}
\end{align}
with $K\in\mathbb{N}$, $\alpha_1,\ldots,\alpha_K:\mathbb{R}^d \rightarrow \mathbb{R}^d$, $\beta_1,\ldots,\beta_K:\mathbb{R}^d \rightarrow \mathbb{R}^{d \times m}$, and 
$\gamma_1,\ldots,\gamma_K$ defined as expectations
\begin{align}
    \gamma_k(t)\coloneqq \mathbb{E}[\varphi_k(X_t)],
    \label{eq:gamma}
\end{align}
where $\varphi_1,\ldots,\varphi_K:\mathbb{R}^d \rightarrow \mathbb{R}$.
Then, assuming that a unique solution of Eq.~\eqref{eq:MVSDE} exists, we aim to estimate the expectation $\mathbb{E}[\phi(X_T)]$ at the terminal time $T$ with a given function $\phi:\mathbb{R}^d\rightarrow\mathbb{R}$.
\label{prob:main}
\end{problem}

We consider this problem under the following assumptions.

\begin{assumption}
     For some $M\in\mathbb{N}$, $\alpha_k$ and $\beta_k$ are bounded and have bounded derivatives up to $M$-th order, namely, there exists $U>1$ such that for any $k\in[K]$, $(l_1,\ldots,l_d)\in\mathbb{N}_0^d$ satisfying $1 \le l_1+\cdots+l_d \le M$, and $x\in \mathbb{R}^d$,
\begin{align}
\left|\partial_{x_1}^{l_1}\cdots\partial_{x_d}^{l_d} \alpha_k(x)\right|\le U, ~ \left|\partial_{x_1}^{l_1}\cdots\partial_{x_d}^{l_d} \beta_k(x)\right|\le U,
\label{eq:alphabetaBound}
\end{align}
holds.
\label{ass:alphabetaBound}
\end{assumption}

\begin{assumption}
$\{\gamma_k\}_k$ are in $C^2([0,T],\mathbb{R})$, and there exists $U>1$ such that for any $k\in[K]$ and $t\in[0,T]$,
\begin{align}
    |\gamma_k(t)|, |\gamma^\prime_k(t)|, |\gamma^{\prime\prime}_k(t)|\le U
    \label{eq:gammaBound}
\end{align}
holds.
\label{ass:gammaBound}
\end{assumption}

\begin{assumption}
For some $M\in\mathbb{N}$, $\varphi_k$'s have bounded derivatives up to $M$-th order, namely, there exists $U>1$ such that for any $k\in[K]$, $(l_1,\ldots,l_d)\in\mathbb{N}_0^d$ satisfying $1 \le l_1+\cdots+l_d \le M$, and $x\in \mathbb{R}^d$,
\begin{align}
\left|\partial_{x_1}^{l_1}\cdots\partial_{x_d}^{l_d} \varphi_k(x)\right|\le U
    \label{eq:basis}
\end{align}
holds.
\label{ass:basis}
\end{assumption}

Note that Assumptions \ref{ass:alphabetaBound} and \ref{ass:basis} immediately imply that $\alpha_k$, $\beta_k$, and $\varphi_k$ are Lipschitz continuous with Lipschitz constants $dU$, $d\sqrt{m}U$, and $\sqrt{d}U$.

Assuming a common upper bound $U$ in Eqs.~\eqref{eq:alphabetaBound} to \eqref{eq:basis} is just for simplicity, and assuming different bounds does not make any essential change in the following discussion. 

Originally, \cite{Belomestny2018} considered the form of the coefficient functions $a$ and $b$ in Eq.~\eqref{eq:MVSDE} in light of orthogonal series expansion.
Let us assume that $a_{\rm MV}$ and $b_{\rm MV}$ in Eq.~\eqref{eq:MVSDEGen}, which depend on the law $\mu_t$, take the forms of
\begin{align}
    a_{\rm MV}(x,\mu_t) = \int \bar{a}(x,u) \mu_t(du),
    b_{\rm MV}(x,\mu_t) & = \int \bar{b}(x,u) \mu_t(du),
\end{align}
with $\bar{a}:\mathbb{R}^d \times \mathbb{R}^d \rightarrow \mathbb{R}^d$ and $\bar{b}:\mathbb{R}^d \times \mathbb{R}^d \rightarrow \mathbb{R}^{d \times m}$ written by
\begin{align}
    \bar{a}(x,u)=\sum_{k=1}^K \alpha_k(x)\varphi_k(u), 
    \bar{b}(x,u)=\sum_{k=1}^K \beta_k(x)\varphi_k(u).
    \label{eq:barabExp}
\end{align}
That is, seen as functions of $u$, $\bar{a}$ and $\bar{b}$ are given as the expansion by the functions $\varphi_k$ with $x$-dependent expansion coefficients $\alpha_k(x)$ and $\beta_k(x)$.
In this setting, the general form of the MVSDE in Eq.~\eqref{eq:MVSDEGen} is boiled down to Eq.~\eqref{eq:MVSDE}.
The form of Eq.~\eqref{eq:barabExp} is naturally motivated if $\varphi_1,\ldots,\varphi_K$ spans an orthonormal system: with some weight function $w:\mathbb{R}^d\rightarrow\mathbb{R}_+$, for any $k,l\in[K]$,
\begin{align}
    \int \varphi_k(x)\varphi_l(x) w(x)dx=
    \begin{cases}
        1~; & k=l \\
        0~; & {\rm otherwise}
    \end{cases}.
    \label{eq:ortho}
\end{align}
The expansion by orthogonal functions is a widely used method for function approximation, and for sufficiently smooth $\bar{a}$ and $\bar{b}$, the expansion like Eq.~\eqref{eq:barabExp} serves as a good approximation.

Although this is the background of the form of Eq.~\eqref{eq:MVSDE}, we do not need to mind this point hereafter: the following discussion holds whether the orthogonal relation \eqref{eq:ortho} is satisfied or not.
Besides, some widely considered models, such as the Shimizu-Yamada model and the Kuramoto-Shinomoto-Sakaguchi model considered later, fall into the form of Eq.~\eqref{eq:MVSDE} with no approximation.

\subsection{Time discretization of SDEs}

When we numerically solve SDEs, not only those of the MV type but also ordinary ones, time $t$ is discretized: we introduce grid points $\{t_i\}_i$ and simulate the random process at the time points.
A simple and widely used discretization method is the Euler method~\cite{Maruyama1955}: for an SDE $dY_t=a(t,Y_t) dt + b(t,Y_t) dW_t$ for the $\mathbb{R}^d$-valued random process $Y_t$, given its value $Y_{t_i}$ at time $t_i$, we approximate $Y_{t_{i+1}}$, the value at the next time point $t_{i+1}$, by
\begin{align}
    Y_{t_{i+1}} \approx Y_{t_i} + a(t_i,Y_{t_i}) h_i + b(t_i,Y_{t_i}) \Delta W_i,
\end{align}
where $h_i\coloneqq t_{i+1}-t_i$, and each component of $\Delta W_i\coloneqq W_{t_{i+1}}-W_{t_i}$ can be sampled from $N(0,h_i)$.
More generally and schematically, a discretization method can be specified by a map $F^h_{a,b,t}$ determined by the coefficient functions $a$ and $b$ in the SDE, time $t$, and the time step size $h$.
That is, we consider a discrete random process $\tilde{Y}_i$ defined as
\begin{align}
    \tilde{Y}_{i+1} = F^{h_i}_{a,b,t_i}(\tilde{Y}_i,Z_i),
    \label{eq:DiscProc}
\end{align}
and approximate $Y_{t_i}$ by $\tilde{Y}_i$.
Here, $Z_i$ is some random variable that we can sample, e.g., $\Delta W_i$ in the Euler method.

The convergence rate of the approximated process $\tilde{Y}_i$ to the true one $Y_t$ in some sense in the limit of small time step is an important property of the discretization method.
In particular, we are now interested in the {\it weak order of accuracy}, which quantifies the convergence rate of expectations evaluated with $\tilde{Y}_i$ in place of $Y_{t_i}$, since we estimate the expectations like Eq.~\eqref{eq:gamma} in solving the MVSDE.
We give the definition of the weak order of accuracy in terms of the local error, namely, the error in one time step, following Ref.~\cite{Rossler01032006}.

\begin{definition}

    Let $p\in\mathbb{R}$.
    Let an $\mathbb{R}^d$-valued random process $Y$ be the solution of an SDE $dY_t=a(t,Y_t) dt + b(t,Y_t) dW_t$, $a:[t_{\rm s},t_{\rm e}]\times\mathbb{R}^d \rightarrow \mathbb{R}^d$, $b:[t_{\rm s},t_{\rm e}]\times\mathbb{R}^d \rightarrow \mathbb{R}^{d \times m}$, defined on an interval $[t_{\rm s},t_{\rm e}]\subset[0,T]$.
    Suppose that for any $h\in[0,t_{\rm e}-t_{\rm s}]$, there exist a map $F^h_{a,b,t_{\rm s}}:\mathbb{R}^d \times \mathcal{S}_h \rightarrow \mathbb{R}^d$ with some set $\mathcal{S}_h$ and a $\mathcal{S}_h$-valued random variable $Z$ that satisfy the following property:
    \begin{itemize}
        \item if we define
        \begin{align}
            \tilde{Y}_h^x\coloneqq F^h_{a,b,t_{\rm s}}(x,Z)
            \label{eq:DiscMethodDef}
        \end{align}
        for $x\in\mathbb{R}$, 
        \begin{align}
            |\mathbb{E}[f(\tilde{Y}_h^x)]-\mathbb{E}[f(Y^{t_{\rm s},x}_{t_{\rm s}+h})]|\le \kappa(1+|x|^{2r}) h^{p+1}
            \label{eq:WeakOrderDef}
        \end{align}
        holds for any $f\in\mathbb{C}^{2(\lceil p \rceil+1)}_P(\mathbb{R}^d,\mathbb{R})$ with some $\kappa\in\mathbb{R}$ and $r\in\mathbb{N}$ that depend not on $h$ but on $f$.
    \end{itemize}
    Then, we say that the method \eqref{eq:DiscMethodDef} is the discretization method for the SDE with weak order $p$.
    \label{def:weakorder}
\end{definition}

The local error bound in Eq.~\eqref{eq:WeakOrderDef} leads to the bound of the total error in using the discretization method as Eq.~\eqref{eq:DiscProc} over the finite time interval, $|\mathbb{E}[f(\tilde{Y}_i)]-\mathbb{E}[f(Y_{t_i})]|=O(h^{p})$.
The definition of the weak order of accuracy with the error bound in this form can be found in the literature, e.g., Refs.~\cite{milstein2013numerical,Rossler01032006}.
To suppress the error in the expectation to $\epsilon$, we can take $h=\Omega(\epsilon^{1/p})$~\footnote{This $\Omega(\cdot)$ is a kind of Landau symbols that denotes the asymptotic lower bound and should not be confused with the sample space $\Omega$.}, which leads to the number of time steps $n_{\rm t}=O(1/\epsilon^{1/p})$.

It is known that the Euler method has weak order 1 in a certain setting~\cite{milstein2013numerical}.
Methods with higher weak order are also known.
Among them are kinds of SRK methods \cite{Rossler01032006,Rossler2010}.
For example, RI1RM proposed in Ref.~\cite{Rossler01032006} can achieve weak order 2 and is applicable to a wide range of SDEs, including those with non-commutative noises.
Following the procedure of the method, we can explicitly write down the formula $F^h_{a,b,t}$.
The random variables $Z$ in this method are discrete ones, which are chosen so that Eq.~\eqref{eq:WeakOrderDef} holds. 
The concrete forms of $F^h_{a,b,t}$ and $Z$ and the condition for achieving the weak order 2 in this method are presented in Appendix~\ref{app:SRK}, along with the definition of non-commutativity.
For more details of SRK methods for the weak approximation, see Ref.~\cite{Rossler01032006}.

\subsection{Particle method for MVSDEs \label{sec:particle}}

We can also apply time discretization methods to an MVSDE.
In this case, we need to estimate the expectations that appear in the SDE.
One way for this is a Monte Carlo-based method called the particle method~\cite{Bossy1997,bossy1996convergence,antonelli2002rate}, in which the expectations are approximated by the sample averages over many generated paths (particles).
The concrete procedure in the case that the SDE takes the form in Eq.~\eqref{eq:MVSDE} is given in Ref.~\cite{Belomestny2018} and shown in Algorithm~\ref{alg:particle}.

\begin{algorithm}[H]
\caption{Particle method for Problem~\ref{prob:main}}\label{alg:particle}
\begin{algorithmic}[1]

\Require
\begin{itemize}
    \item Number of particles $N$
    \item time step size $h$ such that $n_{\rm t} \coloneqq T/h\in\mathbb{N}$
\end{itemize}

\State For $n=1,\ldots,N$, set $\tilde{X}^n_0 \coloneqq X_0$.

\State For $i=0,...,n_{\rm t}$, set $t_i=ih$.

\For{$i=0,...,n_{\rm t}-1$}

\State For $k=1,\ldots,K$, set
\begin{align}
    \hat{\gamma}_{k,i}\coloneqq \frac{1}{N}\sum_{n=1}^N \varphi_k(\tilde{X}^n_i).
    \label{eq:gammaSampleAve}
\end{align} 

\State For $n=1,\ldots,N$, compute
\begin{align}
    &\tilde{X}^n_{i+1} \coloneqq \tilde{X}^n_i + a_i(\tilde{X}^n_i)h + b_i(\tilde{X}^n_i)Z_{n,i}, \nonumber \\
    &a_i(x)\coloneqq\sum_{k=1}^K \hat{\gamma}_{k,i}\alpha_k(x), b_i(x)\coloneqq\sum_{k=1}^K \hat{\gamma}_{k,i}\beta_k(x)
    \label{eq:MVEuler}
\end{align}
where $Z_{n,i}=(Z^1_{n,i},\ldots,Z^m_{n,i})\in\mathbb{R}^m$ with $Z^j_{n,i}$ independently sampled from $N(0,h)$.

\EndFor

\State Output $\frac{1}{N}\sum_{n=1}^N \phi(\tilde{X}^n_{n_{\rm t}})$ as an estimate of $\mathbb{E}[\phi(X_T)]$.

\end{algorithmic}
\end{algorithm}

Ref.~\cite{antonelli2002rate} showed that under some conditions on the SDE, the empirical distribution function of the particles converges to the true one of $X_{t_i}$ in the limit of large $N$ and small $h$ with error of $O\left(\frac{1}{\sqrt{N}}+h\right)$.
Thus, to suppress the error to $\epsilon$, we can take
\begin{align}
N=O\left(\frac{1}{\epsilon^2}\right)    
\label{eq:NPart}
\end{align}
and
\begin{align}
    h=\Omega\left(\epsilon\right),
    \label{eq:hPart}
\end{align}
which implies that
\begin{align}
    n_t=O\left(\frac{1}{\epsilon}\right).
    \label{eq:ntPart}
\end{align}
Therefore, the computational complexity of Algorithm~\ref{alg:particle} in terms of the number of arithmetic operations and sampling $Z_{n,i}$ scales as
\begin{align}
O\left(\frac{1}{\epsilon^3}\right)
\label{eq:CompPart}
\end{align}
with respect to $\epsilon$.

\subsection{Extrapolation of a function \label{sec:extrap}}

In the method proposed later, we estimate $\gamma_k(t)$ for $t\in(t_i,t_{i+1}]$ using its values at the previous time points, $\gamma_k(t_i),\gamma_k(t_{i-1}),\ldots$.
For this, we use the Taylor expansion-based extrapolation.
Namely, assuming that $\gamma_k$ is sufficiently smooth, we approximate $\gamma_k(t)$ for $t\in[t_i, t_{i+1}]$ as
\begin{align}
    \gamma_k(t) \approx \tilde{\gamma}_{k,i}(t) \coloneqq \sum_{l=0}^{r} \frac{1}{l!} \hat{\gamma}^{(l)}_{k,i} (t-t_i)^l.
    \label{eq:Taylorgamma}
\end{align}
Here, $\hat{\gamma}^{(l)}_{k,i}$, an estimate of the $l$-th derivative $\gamma^{(l)}_k(t_i)$ of $\gamma_k$ at $t=t_i$, is calculated with $\gamma_k(t_i),\ldots,\gamma_k(t_{i-s})$, the values of $\gamma_k$ at the latest $s+1$ time grids, where $s \ge l$, by a backward finite difference formula
\begin{align}
    \hat{\gamma}^{(l)}_{k,i}=\frac{1}{h^l} \sum_{j=0}^{s} d^{l}_{s,j} \gamma_k(t_{i-j}),
    \label{eq:finDiff}
\end{align}
where $d^{l}_{s,0},\ldots,d^{l}_{s,s}$ are real constants, and time step sizes $h_i,\ldots,h_{i-l}$ are assumed to be equal to $h$.
Finite difference approximation formulas of derivatives have been proposed, see e.g., \cite{LI200529}.
For example, the two-point backward difference approximation of the first derivative is given by
\begin{align}
    \hat{\gamma}^\prime_{k,i} = \frac{\gamma_k(t_i)-\gamma_k(t_{i-1})}{h_{i-1}}.
\end{align}

The error in this extrapolation method is roughly evaluated as follows, although the rigorous evaluation is presented in the proof of Theorem~\ref{th:main} in Appendix~\ref{app:proof}.
Taylor's theorem tells us that the approximation \eqref{eq:Taylorgamma} has an error of $O(h^{r+1})$, if each $\hat{\gamma}^{(l)}_{k,i}$ is exact.
Now, $\hat{\gamma}^{(l)}_{k,i}$ estimated by the finite difference formula \eqref{eq:finDiff} generally has an error of $O(h^{s-l+1})$~\cite{LI200529}.
Combining these, we see that for $t\in(t_i,t_{i+1}]$, the extrapolation in the form of Eq.~\eqref{eq:Taylorgamma} has an error of $O(h^{r+1}+h^{s+1})$.
For example, for $r=s=1$, which corresponds to the first-order Taylor expansion and the two-point backward difference formula, the error is $O(h^2)$.

When we use the above extrapolation in the current problem, there are some complications, e.g., the absence of $\gamma_k(t_{i-1})$ for the initial time point $t_i=t_0$.
The algorithm proposed in this work addresses such a point, as presented in Sec.~\ref{sec:algo}.

\subsection{Quantum building-blocks}

\subsubsection{Representing numbers on qubits}

In this paper, we consider computation on systems of multiple quantum registers, each of which consists of multiple qubits. 
To represent real numbers, we use the fixed-point binary representation.
Then, for $x\in\mathbb{R}$, we denote by $\ket{x}$ the computational basis state on a quantum register that holds the bit string equal to the binary representation of $x$.
If $x$ is a vector, $x=(x_1,\cdots,x_d)\in\mathbb{R}^d$, $\ket{x}$ is a quantum state on a $d$-register system like $\ket{x}:=\ket{x_1}\cdots\ket{x_d}$.
We assume that every register has a sufficiently large number of qubits and thus neglect errors caused by finite-precision representation.

\subsubsection{Quantum arithmetic circuits}

We can perform arithmetic operations on numbers represented on quantum registers.
For example, we can implement quantum circuits for four basic arithmetic operations such as addition $O_{\mathrm{add}}:\ket{x}\ket{y}\ket{0}\mapsto\ket{x}\ket{y}\ket{x+y}$ and multiplication $O_{\mathrm{mul}}:\ket{x}\ket{y}\ket{0}\mapsto\ket{x}\ket{y}\ket{xy}$, where $x$ and $y$ are integers: see Ref.~\cite{MunosCoreas2022} and the references therein for concrete implementations.
In the finite-precision binary representation, these operations are immediately extended to those for real numbers.
Using circuits like the above, we can construct circuits $A_f$ to approximately compute elementary functions $f$ such as $\exp$ and $\sin$/$\cos$ and functions written as some combination of elementary functions:
\begin{align}
    A_f\ket{x}\ket{0}=\ket{x}\ket{f(x)}
    \label{eq:Uf}
\end{align}
for $x\in\mathbb{R}$ \cite{haner2018optimizing}.

\subsubsection{Loading probability distributions into quantum states}

In this paper, we use quantum circuits to prepare states in which probability distributions are encoded.
That is, for a probability density function $q:\mathbb{R}^d \rightarrow \mathbb{R}_+$ of some random variable, we consider the quantum circuit $U_q$ that acts as
\begin{align}
U_q \ket{0} = \sum_{x \in \mathcal{G}} \sqrt{\bar{q}(x)} \ket{x}.
\label{eq:SPOra}
\end{align}
This may accompany an ancillary register:
\begin{align}
U_q \ket{0}\ket{0} = \sum_{x \in \mathcal{G}} \sqrt{\bar{q}(x)} \ket{x}\ket{\phi_x},
\label{eq:SPOra2}
\end{align}
where $\ket{\phi_x}$ is a quantum state that depends on $x$.
Note that because of the fixed-point binary representation, we need to discretize the density $q$, which is originally continuous, to $\bar{q}:\mathcal{G} \rightarrow [0,1]$ such that $\sum_{x\in\mathcal{G}}\bar{q}(x)=1$, where the finite set $\mathcal{G} \subset \mathbb{R}^d$ is the sufficiently fine grid on $\mathbb{R}^d$.
In this paper, we call $\bar{q}$ the discretization of $q$, and it is given by, for example, integrating $q$ over each cell of the grid $\mathcal{G}$.
To substitute $\bar{q}$ for $q$ in calculating the expectation $E_{q,f} \coloneqq \int f(x)q(x)dx$ of $f:\mathbb{R}^d \rightarrow \mathbb{R}$, $\tilde{E}_{q,f} \coloneqq \sum_{x\in\mathcal{G}} f(x)\bar{q}(x)$ must approximate $E_{q,f}$ well.
Hereafter, assuming that the grid $\mathcal{G}$ is sufficiently fine, we neglect this discretization error and regard $\tilde{E}_{q,f}$ as $E_{q,f}$ for the functions $f$ considered in this paper.

The quantum circuits like Eq.~\eqref{eq:SPOra} have been considered in previous studies, especially for simple but widely used distributions such as the normal distribution.
The Grover-Rudolph method \cite{grover2002creating} is a pioneering one, and other improved methods~\cite{Sanders2019,zoufal2019quantum,Holmes2020,kaneko2022quantum,rattew2022preparing,mcardle2022quantum,MarinSanchez2023,Moosa_2024} have been proposed.
In some methods, the gate cost to construct $U_q$ scales logarithmically with respect to the number of points in $\mathcal{G}$.
This enables us to take an exponentially fine grid as $\mathcal{G}$, which makes regarding $\tilde{E}_{q,f}$ as $E_{q,f}$ reasonable.

Note that if we have the circuit $U_{q_\Theta}$ to load the density $q_\Theta$ of some random variable $\Theta$ as Eq.~\eqref{eq:SPOra}, we can also construct the density loading circuit like Eq.~\eqref{eq:SPOra2} for another random variable $\Theta^\prime$ that is given by $\Theta^\prime=G(\Theta)$ with some function $G$ computable by the arithmetic circuit $A_G$.
That is, by combining $U_{q_\Theta}$ and $A_G$, we can generate $\sum_{\theta \in \mathcal{G}} \sqrt{\bar{q}_\Theta(\theta)} \ket{\theta}\ket{G(\theta)}$, which can be regarded as $\sum_{\theta^\prime \in G(\mathcal{G})} \sqrt{\bar{q}_{\Theta^\prime}(\theta^\prime)} \ket{\theta}\ket{\theta^\prime}$ for sufficiently smooth $G$.
In particular, for a time-discretization method of SDEs with a map $F^h_{a,b,t}$ and a random variable $Z$, if we have the loading circuit $U_{q_Z}$ for $Z$ that generates $\sum_{z \in \mathcal{G}} \sqrt{\bar{q}_Z(z)} \ket{z}$ and the arithmetic circuit $A_{F^h_{a,b,t}}$ for $F^h_{a,b,t}$, we can construct the circuit $U_{F^h_{a,b,t}}$ that acts for any $y\in\mathbb{R}$ as
\begin{align}
    \ket{y}\ket{0}\ket{0}
    & \xrightarrow{U_{q_Z}} \sum_{z \in \mathcal{G}} \ket{y} \sqrt{\bar{q}_Z(z)} \ket{z}\ket{0} \nonumber \\
    & \xrightarrow{A_{F^h_{a,b,t}}} \sum_{z \in \mathcal{G}} \sqrt{\bar{q}_Z(z)} \ket{y} \ket{z} \ket{F^h_{a,b,t}(y,z)}=\ket{y}\sum_{y^\prime\in\mathcal{G}^\prime}\sqrt{\bar{q}_{F^{h}_{a,b,t}}(y^\prime|y)} \ket{z} \ket{y^\prime}.
\end{align}
Here, $\bar{q}_{F^{h}_{a,b,t}}(\cdot|y)$ is a discretization of $q_{F^{h}_{a,b,t}}(\cdot|y)$, the density of $F^{h}_{a,b,t}(y,Z)$, with some grid $\mathcal{G}^\prime$.
That is, with the second register hidden as an ancillary one, $U_{F^h_{a,b,t}}$ acts as
\begin{align}
    U_{F^h_{a,b,t}}\ket{y}\ket{0}=\ket{y}\sum_{y^\prime\in\mathcal{G}^\prime}\sqrt{\bar{q}_{F^{h}_{a,b,t}}(y^\prime|y)} \ket{y^\prime}.
    \label{eq:OraOneStep}
\end{align}
In fact, previous papers \cite{Rebentrost2018,Stamatopoulos2020optionpricingusing,Chakrabarti2021thresholdquantum,kaneko2022quantum} considered such a quantum circuit for the Euler method, which is implemented with a loading circuit for the normal distribution and arithmetic circuits.

In the above discussion, it is assumed that the random variable $Z$ takes continuous values, which necessitates discretization.
If $Z$ is originally discrete as in some kinds of the SRK method, the above discussion proceeds almost similarly, with a slight change in interpretation.
Now, the loading circuit $U_{q_Z}$ for $Z$ generates $\sum_{z \in \mathcal{G}} \sqrt{q_Z(z)} \ket{z}$ with $\mathcal{G}$ and $q_Z$ being the set of the values $Z$ can take and $Z$'s probability mass function, respectively, and $U_{F^h_{a,b,t}}$ acts as $U_{F^h_{a,b,t}}\ket{y}\ket{0}=\ket{y}\sum_{y^\prime\in\mathcal{G}^\prime}\sqrt{q_{F^{h}_{a,b,t}}(y^\prime|y)}\ket{y^\prime}$, where $\mathcal{G}^\prime$ is the set of the values $F^h_{a,b,t}(y,Z)$ can take, and $q_{F^{h}_{a,b,t}}(\cdot|y)$ is its probability mass function.

\subsubsection{Quantum Monte Carlo integration}

For calculating expectations and, more generally, integrals, the quantum algorithm for Monte Carlo integration, which we hereafter call quantum Monte Carlo integration or QMCI, has been proposed~\cite{montanaro2015}, based on quantum amplitude estimation~\cite{brassard2002,suzuki2020amplitude}.
Given the circuits $U_q$ to encode the density $q$ as Eq.~\eqref{eq:SPOra} and $U_f$ to compute the integrand $f$, QMCI outputs an estimate of the expectation $\int f(x)q(x)dx$.
Ref.~\cite{montanaro2015} presents some versions of QMCI, and in this paper, we use the version for integrands with bounded $L_2$ norm.
Leaving the detailed procedure to \cite{montanaro2015}, we now present only the theorem on its query complexity.

\begin{theorem}[Lemma 2.4 in \cite{montanaro2015}, modified]
    Let $\epsilon\in\mathbb{R}_+$ and $\eta\in(0,1)$.
    Suppose that $X$ is a $\mathbb{R}^d$-valued random variable with density $q$.
    Suppose that for a function $f:\mathbb{R}^d\rightarrow \mathbb{R}$, $\|f(X)\|_{2}^2=\int f^2(x) q(x) dx\le U <\infty$ exists with a known upper bound $U>1$.
    Suppose that for $q$ and $f$, we have the quantum circuits $U_q$ and $U_f$ as Eqs.~\eqref{eq:SPOra} and \eqref{eq:Uf}, respectively.
    Then, there is a quantum algorithm $\proc{QMCI}(U_q,U_f,\epsilon,\eta,U)$ that outputs an $\epsilon$-approximation of $\int f(x) q(x) dx$ with probability at least $1-\eta$, querying $U_q$ and $U_f$
    \begin{align}
        O\left(\frac{U}{\epsilon}\log^{3/2}\left(\frac{U}{\epsilon}\right)\log\log\left(\frac{U}{\epsilon}\right)\log\left(\frac{1}{\eta}\right)\right)
        \label{eq:CompQMCI}
    \end{align}
    times.
   \label{th:QMCI}
\end{theorem}

Note that the complexity in Eq.~\eqref{eq:CompQMCI} scales on the accuracy $\epsilon$ as $\widetilde{O}(1/\epsilon)$, which shows the quadratic speed-up compared to classical Monte Carlo integration with sample complexity of $O(1/\epsilon^2)$.

Although Lemma 2.4 in \cite{montanaro2015} states that the success probability of the algorithm is lower bounded by a constant $\frac{4}{5}$, the success probability in the above theorem is enhanced to an arbitrary value $1-\eta$, at the price of an additional factor $\log(1/\eta)$ in the complexity bound in Eq.~\eqref{eq:CompQMCI}.
This is done by running the original algorithm $O(\log(1/\eta))$ times and taking the median of the outputs in the different runs; see Lemma 2.1 in \cite{montanaro2015}, which originates from Lemma 6.1 in \cite{jerrum1986}.

\section{Proposed quantum algorithm \label{sec:main}}

Now, let us present our quantum algorithm for Problem~\ref{prob:main}.

\subsection{Idea \label{sec:idea}}

Before presenting the details, let us give an intuitive explanation on the way to design the quantum algorithm.

In the particle method, the scaling of the particle number $N$ on the accuracy $\epsilon$ as $O(1/\epsilon^2)$ in Eq.~\eqref{eq:NPart} is in line with that of the sample number in general Monte Carlo integration.
The immediate idea for the quantum speed-up is applying QMCI to this problem in some way to reduce the scaling to $O(1/\epsilon)$.
In fact, this seems possible, since Algorithm~\ref{alg:particle} is a sequence of estimating the expectations $\gamma_k(t_i)$ as $\hat{\gamma}_{k,i}$ and simulating the SDE by the Euler method as Eq.~\eqref{eq:MVEuler}.
That is, we easily conceive a quantum version of Algorithm~\ref{alg:particle}, which is, roughly speaking, the iteration of the following steps for $i=0,\ldots,n_{\rm t}-1$: 
\begin{enumerate}
    \item Using the estimates of $\{\gamma_{k,i^\prime}\}_{k\in[K],i^\prime \le i}$ obtained in the previous steps, construct the quantum circuit $U_{\tilde{X}_{i+1}}$ to generate the quantum state that encodes the probability distribution of $\tilde{X}_{i+1}$ defined by Eq.~\eqref{eq:MVEuler} (with $n$ omitted).

    \item Using QMCI with $U_{\tilde{X}_{i+1}}$, estimate $\gamma_{k,i+1}$.
\end{enumerate}
Constructing $U_{\tilde{X}_{i+1}}$ is also possible, if we have $U_{F^h_{a,b,t}}$ acting as Eq.~\eqref{eq:OraOneStep}.
Combining the circuits of this type, each of which corresponds to each step of time evolution by the discretization method like Eq.~\eqref{eq:DiscProc}, we can construct $U_{\tilde{X}_{i+1}}$ (more detailed implementation will be given later, in Sec.~\ref{sec:algo}).

However, this approach has a drawback.
Let us roughly evaluate the total number of queries to $U_{F^h_{a,b,t}}$ as a metric of the complexity of the above approach, since it is iteratively used as a basic component circuit.
To get $\tilde{X}_i$, we iterate the one-step discretized time evolution like Eq.~\eqref{eq:DiscProc} $i$ times, and thus to construct $U_{\tilde{X}_i}$, we make $i$ uses of $U_{F^h_{a,b,t}}$.
In estimating $\gamma_{k,i}$ by QMCI with accuracy $\epsilon$, we make $O(1/\epsilon)$ queries to $U_{\tilde{X}_i}$, in which $U_{F^h_{a,b,t}}$ is queried $O(i/\epsilon)$ times.
Therefore, in estimating $\gamma_{k,i}$ for all $i\in[n_{\rm t}]$, the total number of queries to $U_{F^h_{a,b,t}}$ piles up to $O(n_{\rm t}^2/\epsilon)$, the quadratic scaling on $n_{\rm t}$, the number of time steps.
This is worse than the classical particle method, whose complexity obviously scales as $O(n_{\rm t})$.
This worse scaling of the quantum approach on $n_{\rm t}$ has already been observed in the context of applying QMCI to problems involving time evolution of random processes \cite{doriguello2022,Miyamoto2025Dividing}.
The root of this is that we need to combine $U_{F^h_{a,b,t}}$ for all the time steps in $U_{\tilde{X}_i}$, which corresponds to time evolution from the beginning, namely, the initial value $X_0$.
This is different from the classical Monte Carlo method, in which we store the intermediate values of the random process in memory and use the values at $t_i$ to get those at $t_{i+1}$. 

From this observation, we see that setting $n_{\rm t}=O(1/\epsilon)$ makes the total complexity $O(1/\epsilon^3)$, which is the same as the classical particle method and implies that the above quantum approach does not provide any advantage.

One may think that a solution to this issue is replacing the Euler method-based time evolution as Eq.~\eqref{eq:MVEuler} in the particle method with some higher-order method such as the SRK method that leads to fewer time steps $n_{\rm t}=O(1/\epsilon^{1/p})$.
However, this is not straightforward.
For time evolution from $t_i$ to $t_{i+1}$, the SRK method uses the values of the coefficient functions $a(t,x)$ and $b(t,x)$ in the SDE at $t\in[t_i,t_{i+1}]$.
In the case of MVSDE, we cannot obtain the exact values of $a$ and $b$ and do not have even estimations for $t\in[t_i,t_{i+1}]$ right after time evolution up to $t_i$, with only the estimates $\hat{\gamma}_{k,i^\prime}$ of $\gamma_{k}(t_{i^\prime})$ for $i^\prime \le i$ obtained.

We thus need to extrapolate $\gamma_{k}(t)$ for $t\in[t_i,t_{i+1}]$ using $\{\hat{\gamma}_{k,i^\prime}\}_{i^\prime \le i}$.
A naive way is extrapolating with a constant function, $\gamma_k(t)\approx\hat{\gamma}_{k,i}$, which is indeed what we do in the particle method.
However, this induces the error of $O(h_i)$ in $\gamma_k$ even for smooth $\gamma_k$, since $\gamma_k(t)$ deviates from $\gamma_k(t_i)$ as $\gamma_k(t)-\gamma_k(t_i) \approx \gamma_k^\prime(t_i)(t-t_i)$.
To suppress this error to $\epsilon$, we need to set $h_i$ at most $\Omega(\epsilon)$, which results in $n_{\rm t}=O(1/\epsilon)$, the same order as the case of using the Euler method.
Therefore, we extrapolate $\gamma_k$ by a function of higher degree than a constant one, as described in Sec.~\ref{sec:extrap}.
Concretely, we use a linear function, Eq.~\eqref{eq:Taylorgamma} with $r=1$, which yields the error of $O(h_i^2)$.
Consequently, using a weak-order $p$ discretization method with this extrapolation, we can take $h_i=O(\min\{\epsilon^{1/2},\epsilon^{1/p}\})$ to achieve the accuracy $\epsilon$.
Then, $n_{\rm t}$ can be $O(\max\{1/\epsilon^{1/2},1/\epsilon^{1/p}\})$, and this results in the total complexity of $O(\max\{1/\epsilon^{2},1/\epsilon^{2/p+1}\})$, which shows the advantage over the classical particle method with complexity of $O(1/\epsilon^3)$.
Hereafter, we assume that $p\le 2$, for which we set $h_i=O(\epsilon^{1/p})$.

\subsection{Algorithm \label{sec:algo}}

Now, based on the above idea, we construct our quantum algorithm.
We begin by presenting assumptions on the availability of the higher weak order discretization method and some quantum circuits.

\begin{assumption}
    We have a discretization method with weak order $p\in\left(1,2\right]$ for any SDE that is defined on any interval $[t,u]\subset[0,T]$ and takes the form of Eq.~\eqref{eq:MVSDE} with coefficient functions $a$ and $b$ satisfying Assumptions~\ref{ass:alphabetaBound} and \ref{ass:gammaBound}.
    Besides, for each of such SDEs and any $h\in(0,|u-t|)$, we have an oracle $U_{F^h_{a,b,t}}$ that acts as Eq.~\eqref{eq:OraOneStep} for any $y\in\mathbb{R}$.
    \label{ass:discMethod}
\end{assumption}

\begin{assumption}
    We have access to the oracles $U_{\varphi_k},k\in[K]$ and $U_\phi$ that act as
    \begin{align}
        U_{\varphi_k}\ket{x}\ket{0}=\ket{x}\ket{\varphi_k(x)}, U_{\phi}\ket{x}\ket{0}=\ket{x}\ket{\phi(x)},
    \end{align}
    for any $x\in\mathbb{R}$.
    \label{ass:funcOracle}
\end{assumption}

Then, our algorithm is as shown in Algorithm\ref{alg:main}.

The following is the main theorem in this paper on the accuracy and complexity of Algorithm~\ref{alg:main}.

\begin{theorem}
Let $p\in(1,2]$, $\epsilon\in\mathbb{R}_+$, and $\eta\in(0,1)$.
Suppose that Assumptions \ref{ass:alphabetaBound} to \ref{ass:funcOracle} holds, where $M$ in Assumption \ref{ass:alphabetaBound} satisfies $M\ge 2(\lceil p \rceil + 1)$.
Assume that
\begin{align}
    2\sqrt{(T+m)}dKU^2\exp\left(2(T+m)Td^2K^2U^4\right) \le \frac{1}{12}
    \label{eq:ShortTimeAssum}
\end{align}
holds.
Suppose that, for each $i$, $\tilde{X}_i$ defined as Eq.\eqref{eq:Xtil} has the moments $\mathbb{E}[|\tilde{X}_i|^r]$ uniformly bounded by $U$ with respect to $r\in\mathbb{N}$.
Then, Algorithm~\ref{alg:main} with
\begin{align}
    \epsilon_{\mathrm{QMCI}}=\frac{\epsilon}{12}
    \label{eq:epsQMCI}
\end{align}
\begin{align}
    h_{\rm I}&=\frac{h_{\rm II}}{\left\lceil h_{\rm II}/h_{\rm I,max} \right\rceil},~h_{\rm I,max}=\frac{3\epsilon}{4U}, \nonumber \\
    h_{\rm II}&=\frac{T}{\max\left\{\left\lceil\frac{T}{h_{\rm II,max}}\right\rceil,2\right\}},~h_{\rm II,max}=\left(\frac{\epsilon}{\max\{4U,24\kappa^\prime T\}}\right)^{1/p},
    \label{eq:hIandII}
\end{align}
where $\kappa^\prime$ is a constant defined later (see Appendix~\ref{app:proof}), outputs an $\epsilon$-approximation of $\mathbb{E}[\phi(X_T)]$ with probability at least $1-\eta$, making 
\begin{align}
    O\left(\left(\frac{1}{\epsilon}\right)^{1+\frac{2}{p}}\log^{\frac{3}{2}}\left(\frac{1}{\epsilon}\right)\log\log\left(\frac{1}{\epsilon}\right)\log\left(\frac{1}{\epsilon\eta}\right)\right)
    \label{eq:query1}
\end{align}
queryies to $U_{F^h_{a,b,t}}$ and
\begin{align}
    O\left(\left(\frac{1}{\epsilon}\right)^{1+\frac{1}{p}}\log^{\frac{3}{2}}\left(\frac{1}{\epsilon}\right)\log\log\left(\frac{1}{\epsilon}\right)\log\left(\frac{1}{\epsilon\eta}\right)\right)
    \label{eq:query2}
\end{align}
queries to $U_{\varphi_k}$ and $U_\phi$, where the constant factors independent of $\epsilon$ and $\eta$ are hidden.

    \label{th:main}
\end{theorem}

\newgeometry{bottom=30mm}
\begin{figure}[p]
\begin{algorithm}[H]

\caption{Proposed quantum algorithm for Problem~\ref{prob:main}}\label{alg:main}
\begin{algorithmic}[1]
\small

\Require
\begin{itemize}
    \item time step sizes $h_{\rm I}$ and $h_{\rm II}$ such that $h_{\rm II}/h_{\rm I}\in\mathbb{N}$ and $T/h_{\rm II}-1\in\mathbb{N}$
    \item QMCI accuracy $\epsilon_{\mathrm{QMCI}}$
    \item success probability lower bound $1-\eta$
\end{itemize}

\State Set $n_{\rm t}^{\rm I}=h_{\rm II}/h_{\rm I}$, $n_{\rm t}^{\rm II}\coloneqq T/h_{\rm II}-1$, $n_{\rm t}\coloneqq n_{\rm t}^{\rm I} + n_{\rm t}^{\rm II}$,
\begin{align}
    h_i\coloneqq
    \begin{cases}
        h_{\rm I} & ; ~ i=0,\ldots,n_{\rm t}^{\rm I}-1 \\
        h_{\rm II} & ; ~ i=n_{\rm t}^{\rm I},\ldots,n_{\rm t}
    \end{cases}
    ,
\end{align}
and
\begin{align}
    t_i \coloneqq
    \begin{cases}
        ih_{\rm I} & ; ~ i=0,\ldots,n_{\rm t}^{\rm I} \\
        (i-n_{\rm t}^{\rm I}+1)h_{\rm II} & ; ~ i=n_{\rm t}^{\rm I}+1,\ldots,n_{\rm t}
    \end{cases}
    .
\end{align}

\State Set $\tilde{X}_0 \coloneqq X_0$.

\State For $k=1,\ldots,K$, set $\hat{\gamma}_{k,0}\coloneqq\varphi_k(X_0)$. 

\For{$i=0,...,n_{\rm t}-1$}

\State For $k=1,\ldots,K$, define $\tilde{\gamma}_{k,i}:[0,t_{i+1}]\rightarrow\mathbb{R}$ as
\begin{align}
    \tilde{\gamma}_{k,i}(t)=
    \begin{cases}
        \hat{\gamma}_{k,0}^\prime(t-t_0)+\hat{\gamma}_{k,0} & ; ~ t_0 \le t < t_1 \\
        \qquad \vdots & \\
        \hat{\gamma}_{k,i}^\prime(t-t_i)+\hat{\gamma}_{k,i} & ; ~ t_i \le t \le t_{i+1}
    \end{cases}
    \label{eq:gammatil}
\end{align}
with 
\begin{align}
    \hat{\gamma}_{k,j}^\prime\coloneqq
    \begin{dcases}
        0 & ; ~ j=0,\ldots,n_{\rm t}^{\rm I}-1 \\
        \frac{\hat{\gamma}_{k,j}-\hat{\gamma}_{k,0}}{h_{\rm II}} & ; ~ j=n_{\rm t}^{\rm I} \\
        \frac{\hat{\gamma}_{k,j}-\hat{\gamma}_{k,j-1}}{h_{\rm II}} & ; ~ j=n_{\rm t}^{\rm I}+1,\ldots,n_{\rm t}
    \end{dcases}
    .
\end{align}

\State Using the discretization method $F$, define a $\mathbb{R}^d$-valued discrete random process
\begin{align}
    \tilde{X}_{i+1} \coloneqq F^{h_i}_{a_i,b_i,t_i}(\tilde{X}_i,Z_i).
    \label{eq:Xtil}
\end{align}
Here, $a_i:[0,t_{i+1}]\times\mathbb{R}^d\rightarrow\mathbb{R}^d$ and $b_i:[0,t_{i+1}]\times\mathbb{R}^d\rightarrow\mathbb{R}^{d \times m}$ are defined as $a_i(t,x)=\sum_{k=1}^K \tilde{\gamma}_{k,i}(t)\alpha_k(x)$ and $b_i(t,x)=\sum_{k=1}^K \tilde{\gamma}_{k,i}(t)\beta_k(x)$, respectively, and $Z_i$ is the random variable used in $F$ independent of $Z_0,\ldots,Z_{i-1}$.

\State Construct the quantum circuit $U_{F^{h_i}_{a_i,b_i,t_i}}$, and combine it with the ones constructed in the previous steps as
\begin{align}
    & U_{\tilde{X}_{i+1}} = 
    \left(I^{\otimes i} \otimes U_{F^{h_i}_{a_i,b_i,t_i}}\right)\cdots
    \left(I \otimes U_{F^{h_1}_{a_1,b_1,t_1}}\otimes I^{\otimes(i-1)}\right)
    \left(U_{F^{h_0}_{a_0,b_0,t_0}}\otimes I^{\otimes i}\right)
    \left(U_{X_0}\otimes I^{\otimes (i+1)}\right)
    \label{eq:UtilX}
\end{align}
to get the quantum circuit $U_{\tilde{X}_{i+1}}$ (see Figure~\ref{fig:UXi}).
Here, $U_{X_0}$ is the quantum circuit to set a register to $\ket{X_0}$: $U_{X_0}\ket{0}=\ket{X_0}$, and $I$ is an identity operator on the Hilbert space of each register.

\If{$i \le n_{\rm t}-2$}
\State For $k=1,...,K$, estimate $\mathbb{E}[\varphi_k(\tilde{X}_{i+1})]$ by $\proc{QMCI}(U_{\tilde{X}_{i+1}},U_{\varphi_k},\epsilon_{\rm QMCI},\eta^\prime,U)$, where
\begin{align}
\eta^\prime\coloneqq \frac{\eta}{K(n_{\rm t}-1)+1},
\label{eq:etaPr}
\end{align}
and let the estimate be $\hat{\gamma}_{k,i+1}$.
\Else
\State Estimate $\mathbb{E}[\phi(\tilde{X}_{n_{\rm t}})]$ using $\proc{QMCI}(U_{\tilde{X}_{n_{\rm t}}},U_{\phi},\epsilon_{\rm QMCI},\eta^\prime,U)$, and output the result $\hat{E}$.
\EndIf

\EndFor
\normalsize

\end{algorithmic}
\end{algorithm}
\end{figure}
\restoregeometry

\begin{figure}[t]
\begin{center}

\begin{quantikz}[row sep={0.8cm,between origins}, column sep={0.2cm}]
\lstick{$\ket{X_0}$} & \gate{U_{X_0}} & \gate[2]{U_{F^{h_0}_{a_0,b_0,t_0}}} & \qw  & \qw & \cdots & & \qw & \qw & \rstick{$X_0$} \\
& \qw & & \gate[2]{U_{F^{h_1}_{a_1,b_1,t_1}}} & \qw  & \cdots & & \qw & \qw & \rstick{$\tilde{X}_1$}\\
& \qw & \qw & \qw & \qw & \cdots & & \qw & \qw & \rstick{$\tilde{X}_{2}$}\\
& \vdots & & & & & & & \\
& \qw & \qw & \qw & \qw & \cdots & & \gate[2]{U_{F^{h_i}_{a_i,b_i,t_i}}} & \qw & \rstick{$\tilde{X}_{i}$}\\
& \qw & \qw & \qw & \qw & \cdots & & & \qw& \rstick{$\tilde{X}_{i+1}$}
\end{quantikz}
\caption{The quantum circuit $U_{\tilde{X}_{i+1}}$ to generate the state that encodes the distribution of $\tilde{X}_{i+1}$.}
\label{fig:UXi}
\end{center}
\end{figure}

Leaving the proof to Appendix~\ref{app:proof}, we now make some comments on Algorithm~\ref{alg:main}.
First, note that Algorithm~\ref{alg:main} is divided into two stages with respect to the time step size: in the first stage, $i =1,\ldots, n_{\rm t}^{\rm I}-1$, we use $h_{\rm I}=O(\epsilon)$, and in the second stage, $i =n_{\rm t}^{\rm I},\ldots,n_{\rm t}$, we use $h_{\rm II}=O(\epsilon^{1/p})$.
Although the discussion in Sec.~\ref{sec:idea} motivates us to use the time step size of $O(\epsilon^{1/p})$ in all the steps in the algorithm, there is an issue in the early steps.
That is, for linear extrapolation of $\gamma_k(t)$, we need (estimations of) its values at two time points, but at the initial time $t=0$, we have only $\gamma_k(0)$.
One may think that using constant extrapolation only at the beginning solves the issue.
That is, we may use the constant extrapolation of $\gamma_k(t)$ for the first interval $[0,t_1]$ with $t_1=h_1=O(\epsilon)$ and get $\hat{\gamma}_{k,1}$, an estimate of $\gamma_k(t_1)$, and after that, we may perform linear extrapolation of $\gamma_k(t)$, setting the step size $h_i=O(\epsilon^{1/p})$ at the step $i=1$ and later.
However, extrapolating the values of $\gamma_k(t)$ at points with a small interval of $O(\epsilon)$ to a wider range of $O(\epsilon^{1/p})$ width induces a large extrapolation error, which cannot be suppressed to $O(\epsilon)$.
We thus repeat constant extrapolation with $O(\epsilon)$ step size until the time $t$ becomes $O(\epsilon^{1/p})$.
After that, we increase the step size to $O(\epsilon^{1/p})$ and perform linear extrapolation of $\gamma_k(t)$ using the data points with $O(\epsilon^{1/p})$ interval.
With the width of the extrapolation being of the same order as the data point interval, the extrapolation can be suppressed to $O(\epsilon)$.
Here, it may be possible to take a more sophisticated strategy than repeating constant extrapolation in the first stage, e.g., linear interpolation with a gradually increasing step size.
Now, we simply adopt repeated constant extrapolation, because with this strategy, the first stage has $O(1/\epsilon^{1-1/p})$ steps and makes a subdominant contribution to the algorithm's complexity compared to the second stage, which has $O(1/\epsilon^{1/p})$ steps, for $p\le2$.

Next, let us see that $U_{\tilde{X}_{i+1}}$ in Eq.~\eqref{eq:UtilX} indeed generates a quantum state that encodes the distribution of $\tilde{X}_{i+1}$.
We see this by considering how the sequence of $U_{F^{h_i}_{a_i,b_i,t_i}}$ transforms a quantum state initially set to $\ket{X_0}\ket{0}^{\otimes (i+1)}$.
Combining
\begin{align}
    U_{F^{h_j}_{a_j,b_j,t_j}}\ket{x_{j+1}}\ket{0} = \sum_{x_{j+1}\in\mathcal{G}}  \sqrt{\bar{q}_{F^{h_j}_{a_j,b_j,t_j}}(x_{j+1}|x_j)}  \ket{x_j}\ket{x_{j+1}}
\end{align}
for $j=0,...,i$, we have
\begin{align}
    U_{F^{h_i}_{a_i,b_i,t_i}}\ket{X_0}\ket{0}^{\otimes (i+1)}=\sum_{x_{i+1}\in\mathcal{G}} \ket{\psi_{x_{i+1}}} \ket{x_{i+1}},
\end{align}
where the unnormalized state $\ket{\psi(x_{i+1})}$ is given by
\begin{align}
    & \ket{\psi(x_{i+1})}\coloneqq \nonumber \\
    & \ \sum_{x_1\in\mathcal{G}} \cdots \sum_{x_{i}\in\mathcal{G}} \sqrt{\bar{q}_{F^{h_i}_{a_i,b_i,t_i}}(x_{i+1}|x_i) \cdots \bar{q}_{F^{h_0}_{a_0,b_0,t_0}}(x_1|X_0)}\ket{X_0}\ket{x_1}\cdots\ket{x_i}.
\end{align}
Since
\begin{align}
    \left|~\ket{\psi(x_{i+1})}~\right|^2 & = \sum_{x_1\in\mathcal{G}} \cdots \sum_{x_{i}\in\mathcal{G}}\bar{q}_{F^{h_i}_{a_i,b_i,t_i}}(x_{i+1}|x_i) \cdots \bar{q}_{F^{h_0}_{a_0,b_0,t_0}}(x_1|X_0) \nonumber \\
    & \approx \int dx_1 \cdots \int dx_i q_{F^{h_i}_{a_i,b_i,t_i}}(x_{i+1}|x_i) \cdots q_{F^{h_0}_{a_0,b_0,t_0}}(x_1|X_0)
\end{align}
is the transition probability density from $X_0$ to $\tilde{X}_{i+1}=x_{i+1}$, we see that $U_{\tilde{X}_{i+1}}$ acts as desired in the form of Eq.~\eqref{eq:SPOra2}.

Lastly, we note that the combination of a high weak-order SDE discretization method and an extrapolation of $\gamma_k$ can also be used in the classical particle method.
That is, we can run the particle method, estimating $\gamma_k(t_i)$ as Eq.~\eqref{eq:gammaSampleAve}, extrapolating it into the next time step $[t_i,t_{i+1}]$, and applying the high-order scheme to generate $\tilde{X}^n_{i+1}$, the particles at time $t_{i+1}$.
Since we can take $h$ of order $\epsilon^{1/p}$ again, this approach yields the complexity of $O(1/\epsilon^{2+1/p})$.
Compared to this, the proposed quantum method with $O(1/\epsilon^{1+2/p})$ complexity still has an advantage if $p>1$.

\section{Numerical experiments \label{sec:demo}}

\begin{algorithm}
\caption{Classical emulation of Algorithm~\ref{alg:main}}\label{alg:mainSim}
\begin{algorithmic}[1]

\Require
\begin{itemize}
    \item time step sizes $h_{\rm I}$ and $h_{\rm II}$ such that $h_{\rm II}/h_{\rm I}\in\mathbb{N}$ and $T/h_{\rm II}-1\in\mathbb{N}$
    \item Number of particles $N$
    \item Parameters of \proc{QMCIML}: $N_{\rm G},n_{\rm shot}$
    \item $u_i,l_i$, upper and lower bounds of $\varphi_k(\tilde{X}_{i+1})$ and $\phi(\tilde{X}_{n_{\rm t}})$
\end{itemize}

\State Set $n_{\rm t}^{\rm I},n_{\rm t}^{\rm II},n_{\rm t},\{h_i\}_i,\{t_i\}_i,\{\hat{\gamma}_{k,0}\}_k$ as Algorithm~\ref{alg:main}.

\State For $n=1,\ldots,N$, set $\tilde{X}^n_0 \coloneqq X_0$.

\For{$i=0,...,n_{\rm t}-1$}

\State Define $\{\tilde{\gamma}_{k,i}\}_k$ as Algorithm~\ref{alg:main}.

\State For $n=1,\ldots,N$, compute
\begin{align}
    \tilde{X}_{i+1}^n \coloneqq F^{h_i}_{a_i,b_i,t_i}(\tilde{X}_i^n,Z_i^n).
\end{align}
with $a_i$ and $b_i$ defined as Algorithm~\ref{alg:main} and $\{Z_i^n\}_n$, i.i.d. samples from the same distribution as $Z_i$.

\If{$i \le n_{\rm t}-2$}
\For{$k=1,...,K$}
\State Compute
\begin{align}
    \bar{\gamma}_{k,i+1}\coloneqq \frac{1}{N}\sum_{n=1}^N \max\left\{ \min\left\{\varphi_k(\tilde{X}_{i+1}^n), u_i\right\},l_i\right\}.
\end{align}

\State Run $\proc{QMCIML}\left(\frac{\bar{\gamma}_{k,i+1}-l_i}{u_i-l_i},N_{\rm G},n_{\rm shot},u_i,l_i\right)$ and let the output be $\hat{\gamma}_{k,i+1}^{\rm sc}$.

\State Let $\hat{\gamma}_{k,i+1}\coloneqq\hat{\gamma}_{k,i+1}^{\rm sc}(u_i-l_i)+l_i$
\EndFor

\Else

\State Compute
\begin{align}
    \bar{\phi}\coloneqq \frac{1}{N}\sum_{n=1}^N \max\left\{ \min\left\{\phi(\tilde{X}_{n_{\rm t}}^n), u_i\right\},l_i\right\}.
\end{align}

\State Run $\proc{QMCIML}\left(\frac{\bar{\phi}-l_{n_{\rm t}}}{u_{n_{\rm t}}-l_{n_{\rm t}}},N_{\rm G},n_{\rm shot},u_{n_{\rm t}},l_{n_{\rm t}}\right)$ and let the output be $\hat{\phi}_{\rm sc}$.

\State Output $\hat{E}\coloneqq\hat{\phi}_{\rm sc}(u_{n_{\rm t}}-l_{n_{\rm t}})+l_{n_{\rm t}}$
\EndIf

\EndFor

\end{algorithmic}
\end{algorithm}

\begin{algorithm}
\caption{\proc{QMCIML}, classical emulation of the MLE-based QMCI}\label{alg:QMCIML}
\begin{algorithmic}[1]

\Require
\begin{itemize}
    \item True value of the expectation $\mu\in[0,1]$
    \item Maximum number of applying the Grover operator $N_{\rm G}=2^{M_{\rm G}}, M_{\rm G}\in\mathbb{N}$
    \item Shot number $n_{\rm shot}$
\end{itemize}

\For{$j\in\mathcal{N}\coloneqq\{0,2^0,2^1,\ldots,2^{M_{\rm G}}\}$}

\State Sample $n_{\rm shot}$ values from the Bernoulli distribution with the probability of 1 being $\sin^2((2j+1)\theta_\mu)$, where $\theta_\mu \coloneqq \arcsin(\sqrt{\mu})$, and let the numbers of 0's and 1's be $n_{0,j}$ and $n_{1,j}$, respectively.

\EndFor

\State Find $\hat{\theta}\coloneqq\operatorname*{argmax}_{\theta\in\left[0,\frac{\pi}{2}\right]} \prod_{j\in\mathcal{N}} \sin^{2n_{1,j}}((2j+1)\theta)\cos^{2n_{0,j}}((2j+1)\theta)$ and output $\sin^2\hat{\theta}$.

\end{algorithmic}
\end{algorithm}

Finally, to demonstrate the validity of our algorithm, we conduct numerical experiments.
Because our algorithm is assumed to run on a fault-tolerant quantum computer that is not available today, we conduct a kind of classical emulation of Algorithm~\ref{alg:main}.
It is a Monte Carlo-based one like the particle method: we generated many paths (particles) of $\tilde{X}_i$ using the discretization method $F$ and the same extrapolation of $\gamma_k$ as Algorithm~\ref{alg:main}.
The estimation of $\gamma_k(t_i)$ by QMCI is replaced with the following.
Setting the number of particles as large as possible, we substitute $\bar{\gamma}_{k,i}$, the average of $\varphi_k(\tilde{X}_i)$ over the particles, for the true value of $\gamma_k(t_i)$.
Then, to emulate the QMCI error, we sample a value from the probability distribution of the output of QMCI applied to an expectation estimation with the true value $\bar{\gamma}_{k,i}$, and let the sampled value be $\hat{\gamma}_{k,i}$, an estimate of $\gamma_k(t_i)$ used in extrapolation.

The entire procedure is described as Algorithm~\ref{alg:mainSim}.
Here, although Algorithm~\ref{alg:main} uses \proc{QMCI}, the QMCI for integrands with bounded $L_2$ norm adopted from Ref.~\cite{montanaro2015}, we replace it with \proc{QMCIML}, the maximum likelihood estimation (MLE)-based QMCI \cite{suzuki2020amplitude}, for which sampling from the output distribution is much simpler.
The procedure for the sampling is given as Algorithm~\ref{alg:QMCIML}.
When MLE-based QMCI is run on quantum hardware, 0 or 1 in step 2 is obtained through a measurement on a quantum state, whose preparation involves $2j+1$ queries to the density loading circuit.
Taking $N_{\rm G}$ of order $1/\epsilon$, for which the total query number is $O(1/\epsilon)$, MLE-based QMCI yields an $\epsilon$-approximation of the expectation with high probability. 
Note that in expectation calculations in Algorithm~\ref{alg:mainSim}, we clip the integrand between $u_i$ and $l_i$ to make the situation match \proc{QMCIML}, in which the integrand is assumed to be bounded.
We expect that this handling, though it causes another error, suffices for the current purpose of validating our algorithm, with $u_i$ and $l_i$ set appropriately.

\subsection{Case 1: Shimizu-Yamada model}

First, we apply Algorithm~\ref{alg:mainSim} to Problem~\ref{prob:main} with
\begin{align}
    & d=1, K=2, \nonumber \\
    & \varphi_1(x) = 1,~\varphi_2(x)=\phi(x)=x, \nonumber \\
    & \alpha_1(x)=0,~\alpha_2(x)=-1, \nonumber \\
    & \beta_1(x)=1,~\beta_2(x)=0, \nonumber \\
    & T = 2, X_0=1,
    \label{eq:ProbParam}
\end{align}
which corresponds to the MVSDE
\begin{align}
    dX_t = -\mathbb{E}[X_t]dt+dW_t.
\end{align}
This is a specific case of the Shimizu-Yamada model (Example 4.2 in Ref.~\cite{Belomestny2018}), which originates from Ref.~\cite{ShimizuYamada1972}.
Ref.~\cite{Belomestny2018} showed that in this case, the MVSDE can be analytically solved as
\begin{align}
    X_t=X_0 e^{-t}+W_t,
\end{align}
for which
\begin{align}
    \gamma_2(t)=\mathbb{E}[\varphi_2(X_t)]=\mathbb{E}[\phi(X_t)]=\mathbb{E}[X_t]=X_0e^{-t}.
\end{align}
Trivially, $\gamma_1(t)=1$.
We aim to see whether Algorithm~\ref{alg:mainSim}, a proxy for Algorithm~\ref{alg:main}, works as expected via comparing its numerical result with this analytical formula.
The parameters in Algorithm~\ref{alg:mainSim} are set as\footnote{For a real number $x$, $\left\lceil x+\frac{1}{2} \right\rceil$ is the result of rounding $x$ to the nearest integer.}
\begin{align}
    & h_{\rm I}=\frac{T}{9}\varepsilon, h_{\rm II}=\frac{T}{3}\sqrt{\varepsilon}, N_{\rm G} = 2^{M_{\rm G}},M_{\rm G}=\left\lceil\log_2 \left(\frac{32}{\varepsilon}\right)+\frac{1}{2}\right\rceil, \nonumber \\
    & N=10^6, n_{\rm shot}=30,\nonumber \\
    & u_i=X_0e^{\alpha t_i}+5\sqrt{t_i},l_i=X_0e^{\alpha t_i}-5\sqrt{t_i}.
    \label{eq:AlgoParam}
\end{align}
Here, $h_{\rm I}$, $h_{\rm II}$, and $N_{\rm G}$ are set with an auxiliary parameter $\varepsilon$ so that they scale as $O(\varepsilon)$, $O(\sqrt{\varepsilon})$, and $O(1/\varepsilon)$, respectively.
This setting reflects the one in Algorithm~\ref{alg:main} with $p=2$, where we set $h_{\rm I}=O(\epsilon)$, $h_{\rm II}=O(\sqrt{\epsilon})$, and $\epsilon_{\rm QMCI}=O(\epsilon)$ to achieve the accuracy $\epsilon$, and what we aim to numerically check is that the error in Algorithm~\ref{alg:mainSim} scales as $O(\varepsilon)$ with this setting.
$N$ is set as large as possible, so that the Monte Carlo statistical error is subdominant compared to the QMCI error and the time discretization error.
$u_i$ and $l_i$ corresponds to the $5\sigma$ range of $X_{t_i}$.
As the SDE discretization method, we use RI1 in \texttt{DifferentialEquations.jl}~\cite{rackauckas2017differentialequations} with $p=2$, an implementation of RI1RM in \cite{Rossler01032006}.

The result is as follows.
FIG.~\ref{fig:RMSE} shows the error of Algorithm~\ref{alg:mainSim} versus $\varepsilon$, which is evaluated as the root mean square deviation of the outputs $\hat{E}$ from the true value of $\mathbb{E}[X_T]$ in 10 runs of the algorithm.
The error decreases as $\varepsilon$ decreases with scaling as $O(\varepsilon)$, which indicates that the algorithm is working as expected.
FIG.~\ref{fig:uMeanvst} shows the values of $\hat{\gamma}_{2,i}$, the estimate of $\gamma_2(t_i)=\mathbb{E}[X_{t_i}]$ by Algorithm~\ref{alg:mainSim} with $\varepsilon=1/12$.
They closely fit the exact values, which also indicates that Algorithm~\ref{alg:mainSim} works accurately. 
Viewing Algorithm~\ref{alg:mainSim} as a proxy of Algorithm~\ref{alg:main}, we plot its computational complexity versus the realized error in FIG.~\ref{fig:Query}.
Here, the complexity is measured as the total number of queries to $U_{F^h_{a,b,t}}$, supposing that $\proc{QMCIML}$ in Algorithm~\ref{alg:mainSim} were run as a quantum algorithm that queries $U_{F^h_{a,b,t}}$.
The complexity is roughly proportional to the inverse square of the error, which is consistent with Theorem~\ref{th:main}.

\begin{figure}
\centering
\begin{minipage}{\linewidth}
    \centering
    \includegraphics[width=0.59\linewidth]{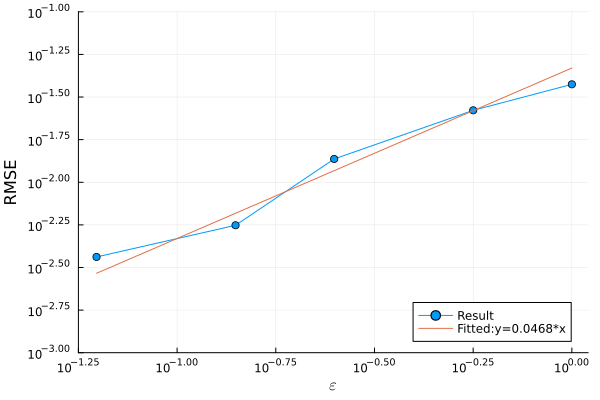}
    \caption{The blue curve represents the root mean square deviation of the outputs $\hat{E}$ from the true value of $\mathbb{E}[X_T]$ in 10 runs of Algorithm~\ref{alg:mainSim} for the Shimizu-Yamada model, where the parameters are set as Eqs.~\eqref{eq:ProbParam} and \eqref{eq:AlgoParam} with varying $\varepsilon$. The red line is the slope-1 line fitted to the blue curve using least squares in the log-log space.}
\label{fig:RMSE}
\end{minipage}

\vspace{0.5em}

\begin{minipage}{\linewidth}
    \centering
    \includegraphics[width=0.59\linewidth]{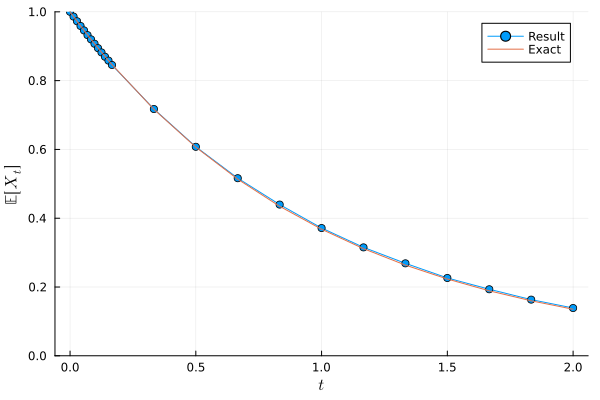}
    \caption{$\mathbb{E}[X_{t}]$ estimated by Algorithm~\ref{alg:mainSim} with $\varepsilon=1/12$ (blue) and its exact value (red) for the Shimizu-Yamada model.}
    \label{fig:uMeanvst}
\end{minipage}

\vspace{0.5em}

\begin{minipage}{\linewidth}
    \centering
    \includegraphics[width=0.59\linewidth]{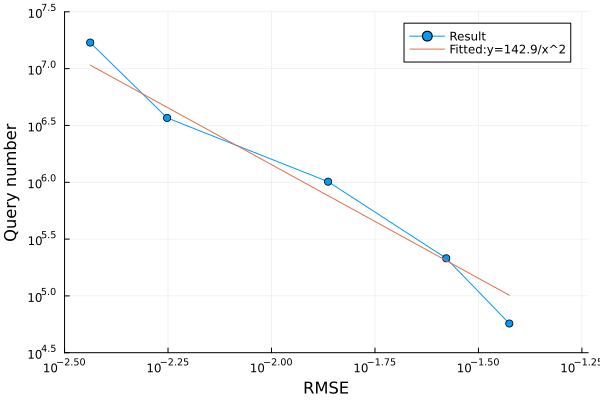}
    \caption{The blue curve represents the total number of queries to $U_{F^h_{a,b,t}}$ in Algorithm~\ref{alg:mainSim} for the Shimizu-Yamada model versus the realized error in the algorithm, which is plotted in FIG.~\ref{fig:RMSE}. Here, we assume that $\proc{QMCIML}$ in it were run as a quantum algorithm that queries $U_{F^h_{a,b,t}}$. The red one represents the function $y=a/x^2$ with $a$ tuned via the least-squares fitting to the blue curve in the log-log space.}
    \label{fig:Query}
\end{minipage}
\end{figure}

\subsection{Case 2: Kuramoto-Shinomoto-Sakaguchi model with a stochastic volatility}

The second example is a multidimensional problem with a non-commutative noise:
\begin{align}
    & d=2, K=3, \nonumber \\
    & \varphi_1(\theta,\sigma) = 1,~\varphi_2(\theta,\sigma)=\phi(\theta,\sigma)=\sin \theta,~\varphi_3(\theta,\sigma)=\cos \theta, \nonumber \\
    & \alpha_1(\theta,\sigma)=\begin{pmatrix}0 \\ 0\end{pmatrix},~\alpha_2(\theta,\sigma)=\begin{pmatrix}\cos \theta \\ 0\end{pmatrix},~\alpha_3(\theta,\sigma)=\begin{pmatrix}-\sin \theta \\ 0\end{pmatrix}, \nonumber \\
    & \beta_1(\theta,\sigma)=\begin{pmatrix}\sigma & 0 \\ 0 & 0.25\sigma\end{pmatrix}, ~\beta_2(x)=\beta_3(x)=\begin{pmatrix}0 & 0 \\ 0 & 0\end{pmatrix}, \nonumber \\
    & T = 5, X_0=\begin{pmatrix}1 \\ 0.5\end{pmatrix},
    \label{eq:ProbParamK}
\end{align}
which corresponds to the following MVSDE for $X_t=(\theta_t,\sigma_t)$:
\begin{align}
    d\theta_t &= \int \sin(u-\theta_t) \mu_t(du) \cdot dt+\sigma_t dW^1_t, \nonumber \\
    d\sigma_t &= 0.25 \sigma_t dW^2_t.
\end{align}
This is an extension of the Kuramoto-Shinomoto-Sakaguchi (KSS) model~\cite{Shinomoto1986,Sakaguchi1988}, which is widely used to model synchronization phenomena in fields such as neuroscience and biophysics, by promoting $\sigma$, which is usually a constant, to a stochastic process $\sigma_t$.

We apply Algorithm~\ref{alg:mainSim} to this problem with the same setting as Eq.~\eqref{eq:AlgoParam} except $u_i$ and $l_i$ set to $\pm 1$.
The result is shown in Figs.~\ref{fig:RMSEK}, \ref{fig:uMeanvstK}, and \ref{fig:QueryK}, which are similar to Figs.~\ref{fig:RMSE}, \ref{fig:uMeanvst}, and \ref{fig:Query} for the Shimizu-Yamada model.
Now, we do not have the analytic solution of the MVSDE, and so we use the classical particle method with large $N$ and small $h$ as a benchmark.
Concretely, we run Algorithm~\ref{alg:particle} with $N=10^6$ and $h=T/256$ and use the sample average of $\phi(\theta_T,\sigma_T)$ over the particles in place of the exact expectation. 
Again, the figures indicate that the proposed method is working as expected: the error decreases roughly linearly with respect to $\varepsilon$, the estimated expectation almost coincides with the benchmark, and the complexity scales roughly proportional to the inverse square of the error.

\begin{figure}
\centering
\begin{minipage}{\linewidth}
    \centering
    \includegraphics[width=0.59\linewidth]{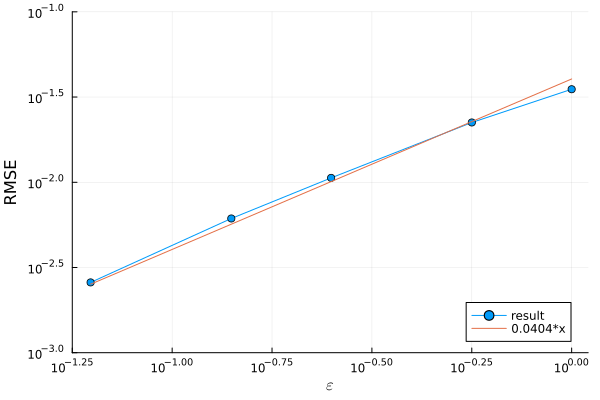}
    \caption{The blue curve represents the root mean square deviation of the outputs $\hat{E}$ from the benchmark value of $\mathbb{E}[\sin(\theta_T)]$ in 10 runs of Algorithm~\ref{alg:mainSim} for the KSS model, where the parameters are set as Eqs.~\eqref{eq:ProbParamK} and \eqref{eq:AlgoParam} (except $u_i$ and $l_i$ set to $\pm1$) with varying $\varepsilon$. The benchmark is the Euler discretization-based particle method with $N=10^6$ and $h=T/256$. The red line is the slope-1 line fitted to the blue curve using least squares in the log-log space.}
\label{fig:RMSEK}
\end{minipage}

\vspace{0.5em}

\begin{minipage}{\linewidth}
    \centering
    \includegraphics[width=0.59\linewidth]{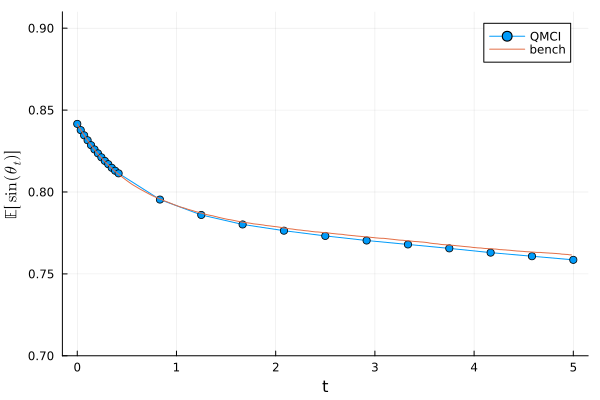}
    \caption{$\mathbb{E}[\sin(\theta_{t})]$ estimated by Algorithm~\ref{alg:mainSim} with $\varepsilon=1/12$ (blue) and the benchmark value (red) for the KSS model.}
    \label{fig:uMeanvstK}
\end{minipage}

\vspace{0.5em}

\begin{minipage}{\linewidth}
    \centering
    \includegraphics[width=0.59\linewidth]{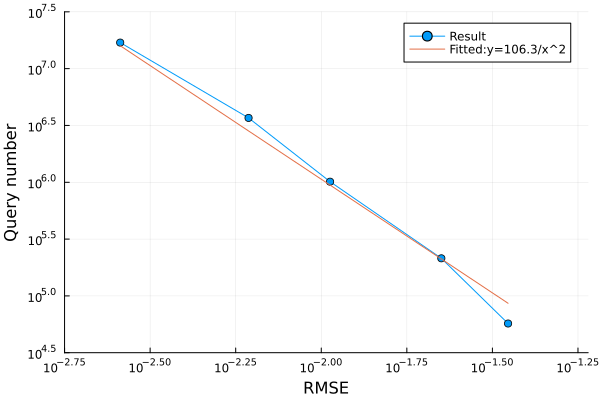}
    \caption{The blue curve represents the total number of queries to $U_{F^h_{a,b,t}}$ in Algorithm~\ref{alg:mainSim} for the KSS model versus the realized error in the algorithm, which is plotted in FIG.~\ref{fig:RMSEK}. The query number is counted similarly to Fig.~\ref{fig:Query}. The red one represents the function $y=a/x^2$ with $a$ tuned via the least-squares fitting to the blue curve in the log-log space.}
    \label{fig:QueryK}
\end{minipage}
\end{figure}

\section{Summary \label{sec:sum}}

In this paper, we considered, for the first time, applying QMCI to MVSDEs.
Taking MVSDEs with the coefficient functions depending on the law of the solution through some expectations $\gamma_k(t)=\mathbb{E}[\varphi_k(X_t)]$ as a first target, we consider the quantum approach in which $\gamma_k(t)$ is estimated by QMCI.
To address the quadratic complexity scaling of the quantum approach on the number of time steps, which is worse than the classical method, we consider time evolution with higher-order discretization method such as the SRK method and linear extrapolation of $\gamma_k(t)$, which allows us to take a time step size $h$ of $O(\epsilon^{1/p})$.
Combining this with QMCI, we devised a quantum algorithm that outputs an $\epsilon$-approximation of the expectation of the given function at the terminal time with $O(1/\epsilon^{1+2/p})$ complexity, which shows an improvement compared to $O(1/\epsilon^{3})$ complexity of the classical method.
To confirm this performance of our algorithm, we conducted numerical demonstrations on the Shimizu-Yamada model and the KSS model with a stochastic volatility, where we saw that the accuracy and complexity behave as expected.

Our algorithm may be improved in some ways.
A straightforward improvement is to use a higher-order scheme---quadratic, cubic, and so on---for extrapolating $\gamma_k(t)$ than the currently adopted linear one.
The current linear extrapolation induces an $O(h^2)$ error, which dominates the SDE time discretization error of $O(h^p)$ in a discretization method with order $p>2$, thereby discouraging the use of such methods. 
Combining a higher-order extrapolation method and SDE discretization method can allow us to take larger $h$ for a given accuracy, making the algorithm more efficient.
A possible issue is that when we take a large step size in the second stage, the first stage involving constant extrapolation with a smaller step size may be a bottleneck.
We may adopt a more sophisticated approach in the first stage, e.g., gradually increasing the extrapolation order from constant to higher as time steps progress.
We will study such an advanced procedure in future works.

Besides, in future works, to explore the practical utility of the proposed quantum algorithm, we would consider its applications to more concrete problems involving MVSDEs in various fields such as finance and fluid mechanics.

\section*{Abbreviations}

QMCI, quantum Monte Carlo integration; SDE, stochastic differential equation; MVSDE, McKean-Vlasov stochastic differential equation; SRK, stochastic Runge-Kutta; MLE, maximum likelihood estimation; KSS, Kuramoto-Shinomoto-Sakaguchi.

\section*{Declaration}

\subsection*{Availability of data and materials}

The code and data for this study can be obtained at https://github.com/Koichi-Miyamoto/MVQMCI.

\subsection*{Competing interests}

The author declares no competing interests.

\subsection*{Funding}

The author is supported by MEXT Quantum Leap Flagship Program (MEXT Q-LEAP) Grant no. JPMXS0120319794 and JST COI-NEXT Program Grant No. JPMJPF2014.

\subsection*{Authors' contributions}

KM as the sole author of the manuscript, conceived, designed, and performed the analysis; he also wrote and reviewed the paper. The author read and approved the final manuscript.

\section*{Acknowledgements}

Not applicable.

\begin{appendices}

\section{Stochastic Runge-Kutta methods for weak approximation \label{app:SRK}}

Here, we present a brief explanation of SRK methods for weak approximation, leaving more details to Ref.~\cite{Rossler01032006}.

Consider an Ito SDE $dY_t=a(t,Y_t) dt + b(t,Y_t) dW_t$ defined on an interval $[t_0,t_1]$ with coefficient functions $a:[t_0,t_1]\times\mathbb{R}^d \rightarrow \mathbb{R}^d$, $b:[t_0,t_1]\times\mathbb{R}^d \rightarrow \mathbb{R}^{d \times m}$.
Suppose that the $\mathbb{R}^{d}$-valued random process $Y_t$ has a deterministic initial value $Y_{t_0}=y\in\mathbb{R}^d$.
For $h\in[0,t_1-t_0]$, the $s$-stage SRK method for $Y_t$ in a general form is given by
\begin{align}
    &F^h_{a,b,t_0}(y,Z)=y+\sum_{i=1}^{s} z^{(0,0)}_i a\left(t_0+c^{(0,0)}_i h, H^{(0,0)}_{i}\right) \nonumber \\
    &\qquad \qquad \qquad +\sum_{i=1}^{s}\sum_{k=1}^{m}\sum_{\nu\in \mathcal{M}}z^{(k,\nu)}_i b^{k}\left(t_0+c^{(k,\nu)}_i h,H^{(k,\nu)}_{i}\right)
    \label{eq:SRK}
\end{align}
with
\begin{align}
    &H^{(k,\nu)}_{i}=y+\sum_{j=1}^{s} z^{(k,\nu),(0,0)}_{ij} a\left(t_0+c^{(0,0)}_j h,H^{(0,0)}_{j}\right) \nonumber \\
    & \qquad \qquad +\sum_{j=1}^{s}\sum_{r=1}^{m}\sum_{\mu\in \mathcal{M}} z^{(k,\nu),(r,\mu)}_{ij} b^{r}\left(t_0+c^{(r,\mu)}_j h,H^{(r,\mu)}_{j}\right) \nonumber \\
    & z^{(k,\nu)}_i=
    \begin{cases}
        \alpha_{i} h & ; ~ k=0,\nu=0 \\
        \sum_{\iota\in \mathcal{M}}{\gamma^{(\iota)}_{i}}^{(k,\nu)} \theta_{\iota}(h) & ; ~ k\in[m],\nu\in\mathcal{M}
    \end{cases}
    \nonumber \\
    & z^{(k,\nu),(r,\mu)}_{ij} =
    \begin{cases}
        A^{(k,\nu),(0,0)}_{ij} h & ; ~ {\rm if } ~ r=0,\mu=0 \\
        \sum_{\iota\in \mathcal{M}}{B^{(\iota)}_{ij}}^{(k,\nu),(r,\mu)} \theta_{\iota}(h) & ; ~ {\rm if } ~ r\in[m],\mu\in\mathcal{M}
    \end{cases}
    \nonumber \\
    & c^{(k,\nu)}_i=\sum_{j=1}^s A^{(k,\nu),(0,0)}_{ij},
\end{align}
for $i,j\in[s]$ and $(k,\nu)\in \{(0,0)\} \cup [m]\times\mathcal{M}$.
Here, $b^k$ is the $k$-th column of $b$, and $\alpha_i$, ${\gamma^{(\iota)}_{i}}^{(k,\nu)}$, $A^{(k,\nu),(0,0)}_{ij}$, and ${B^{(\iota)}_{ij}}^{(k,\nu),(r,\mu)}$ are real constant parameters.
$\mathcal{M}$ is a finite set of multi-indices with size $|\mathcal{M}|=\kappa$.
$\theta_\nu(h),\nu\in\mathcal{M}$ are some random variables dependent on $h$, collectively denoted by $Z$, and satisfy $\theta_\nu(0)=0$ and
\begin{align}
    \mathbb{E}\left[\left(\theta_{\nu_1}(h)\right)^{p_1} \cdots \left(\theta_{\nu_\kappa}(h)\right)^{p_\kappa}\right]=O(h^{(p_1+\cdots+p_\kappa)/2})
\end{align}
for any $p_1,\ldots,p_\kappa\in\mathbb{N}_0$ and $\nu_1,\ldots,\nu_\kappa\in\mathcal{M}$.
Roughly speaking, this means that $\theta_\nu(h)=O(h^{1/2})$.
We also define $\theta_0(h)=h$.

Let us define $\tilde{Y}_{t_0+h} \coloneqq F^h_{a,b,t_0}(y,Z)$ and consider $f(\tilde{Y}_{t_0+h})$ for some $f:\mathbb{R}^d \rightarrow \mathbb{R}$.
Assuming that $a$, $b$, and $f$ are sufficiently smooth, plugging Eq.~\eqref{eq:SRK} into $f(\tilde{Y}_{t_0+h})$ and performing the Taylor expansion leads to an expression of $f(\tilde{Y}_{t_0+h})$ as a series with each term being a product of $z^{(k,\nu)}_i$, $z^{(k,\nu),(r,\mu)}_{ij}$, and a deterministic number.
This yields a numerous number of terms, but we can cleanly write down the series using the stochastic rooted tree analysis.
We consider tree graphs, each of which consists of nodes that are labeled by positive integers and take $m+2$ types: the root node $\otimes$, the deterministic node $\CIRCLE$, and stochastic nodes $\Circle_j$ for $j\in[m]$ (see Fig.~\ref{fig:tree} as examples).
Generating trees and associating some quantities with each tree under a certain rule, we obtain the following formula of $f(\tilde{Y}_{t_0+h})$: for $n$
\begin{align}
    f(\tilde{Y}_{t_0+h})=\sum_{\substack{\mathbf{t} \in LTS(\Delta) \\ l(\mathbf{t}) \le n+1}} \sum_{j_1,\ldots,j_{s(\mathbf{t})}=1}^m \frac{\gamma(\mathbf{t}) \cdot \Phi_S(\mathbf{t}) \cdot F(\mathbf{t})(y)}{(l(\mathbf{t})-1)!}+\mathcal{R}_n,
    \label{eq:RKExp}
\end{align}
Here, $\mathcal{R}_n$ is given by
\begin{align}
    \mathcal{R}_n(t_0+h,t_0)\coloneqq \frac{1}{(n+1)!} \sum_{\nu_1,\ldots,\nu_{n+1}\in\{0\}\cup\mathcal{M}} \theta_{\nu_1}(h)\cdots\theta_{\nu_{n+1}}(h) \frac{\partial^{n+1} f(\tilde{Y}_{t_0+\xi h})}{\partial \theta_{\nu_1}\cdots\partial \theta_{\nu_{n+1}}}
\end{align}
with some $\xi\in(0,1)$, and thus is of $O(h^{(n+1)/2})$.
$LTS(\Delta)$ is the set of trees generated by the rule.
For tree $\mathbf{t}$, $l(\mathbf{t})$ (resp. $s(\mathbf{t})$) is the number of nodes (resp. stochastic nodes) in it, $\Phi_S(\mathbf{t})$ is a random variable defined with $z^{(k,\nu)}_i$ and $z^{(k,\nu),(r,\mu)}_{ij}$, $F(\mathbf{t}):\mathbb{R}^d\rightarrow\mathbb{R}$ is defined with $f$, $a$, $b$, and their derivatives, and $j_1,\ldots,j_{s(\mathbf{t})}$ are $s(\mathbf{t})$ indices of stochastic nodes in $\mathbf{t} \in LTS(\Delta)$.
For their exact definitions, see \cite{Rossler01032006}.

\begin{figure}[t]
\begin{center}
\begin{minipage}[b]{0.49\columnwidth}
    \centering
\begin{tikzpicture}
    \node[draw, shape=circle,minimum width=0.1,label=left:1] (v1) at (0,0) {};
    \node[draw, shape=circle,minimum width=1em, fill=black,label=left:2] (v2) at (-1,1) {};
    \node[draw, shape=circle,minimum width=1em,label=left:3] (v3) at (1,1) {};
    \node[draw, shape=circle,minimum width=1em,label=left:4] (v4) at (-1,2) {};
    \node[below right=-0.15cm of v3] {\footnotesize{$j_1$}};
    \node[below right=-0.15cm of v4] {\footnotesize{$j_2$}};

    \draw [thick]
    (v1) -- (v2) node {}
    (v1) -- (v3) node {}
    (v2) -- (v4) node {};
    \draw [thick] (-0.11,-0.11) -- (0.11,0.11);
    \draw [thick] (-0.11,0.11) -- (0.11,-0.11);
\end{tikzpicture}
\end{minipage}
\begin{minipage}[b]{0.49\columnwidth}
    \centering
\begin{tikzpicture}
    \node[draw, shape=circle,minimum width=0.1,label=left:1] (v1) at (0,0) {};
    \node[draw, shape=circle,minimum width=1em,label=left:2] (v2) at (0,1) {};
    \node[draw, shape=circle,minimum width=1em,label=left:3] (v3) at (1,2) {};
    \node[draw, shape=circle,minimum width=1em, fill=black,label=left:4] (v4) at (-1,2) {};
    \node[below right=-0.15cm of v2] {\footnotesize{$j_1$}};
    \node[below right=-0.15cm of v3] {\footnotesize{$j_2$}};

    \draw [thick]
    (v1) -- (v2) node {}
    (v2) -- (v3) node {}
    (v2) -- (v4) node {};
    \draw [thick] (-0.11,-0.11) -- (0.11,0.11);
    \draw [thick] (-0.11,0.11) -- (0.11,-0.11);
\end{tikzpicture}
\end{minipage}

\caption{Two examples of stochastic rooted trees with $j_1,j_2\in[m]$, taken from \cite{Rossler01032006}.}
\label{fig:tree}
\end{center}
\end{figure}

On the other hand, it is known that for $p\in\mathbb{N}_0$ and sufficiently smooth $f$, $a$, and $b$, the following expansion holds:
\begin{align}
    \mathbb{E}[f(Y_{t_0+h})]=\sum_{\substack{\mathbf{t} \in LTS(I) \\ \rho(\mathbf{t}) \le p}} \sum_{j_1,\ldots,j_{s(\mathbf{t})/2}=1}^m \frac{F(\mathbf{t})(y)}{2^{s(\mathbf{t})/2}\rho(\mathbf{t})!}h^{\rho(\mathbf{t})}+O(h^{p+1}).
    \label{eq:STaylor}
\end{align}
Here, $LTS(I)$ is a certain set of trees, and for tree $\mathbf{t}$, $\rho(\mathbf{t})=d(\mathbf{t})+\frac{1}{2}s(\mathbf{t})$ with $d(\mathbf{t})$ being the number of deterministic nodes, and $j_1,\ldots,j_{s(\mathbf{t})/2}$ are $s(\mathbf{t})/2$ different indices of stochastic nodes in $\mathbf{t} \in LTS(I)$.

Given the above facts, the strategy for designing a SRK method with weak order $p$ is as follows: based on Eqs.~\eqref{eq:RKExp} and \eqref{eq:STaylor}, we choose the parameters $\alpha_i$, ${\gamma^{(\iota)}_{i}}^{(k,\nu)}$, $A^{(k,\nu),(0,0)}_{ij}$, and ${B^{(\iota)}_{ij}}^{(k,\nu),(r,\mu)}$, and random variables $Z$ so that $\mathbb{E}[f(\tilde{Y}_{t_0+h})]$ agrees with $\mathbb{E}[f(Y_{t_0+h})]$ up to terms of $O(h^{p+1})$.
In fact, some kinds of SRK methods have been proposed.
One example is RI1WM proposed in \cite{Rossler01032006}.
It has three stages and weak order $2$.
$\mathcal{M}$ is set as $\mathcal{M}=\{\{j_1\},\{j_1,j_2\} : j_1,j_2\in[m]\}$.
The parameters are specified with the Butcher array
\renewcommand{\arraystretch}{2}
\begin{equation*}
\begin{tabular}{c|c|c|c}
    $c^{(0,0)}$ & ${A}^{(0,0),(0,0)}$ & ${B^{(k)}}^{(0,0),(k,k)}$ & \\
    \hline
    $c^{(k,k)}$ & ${A}^{(k,k),(0,0)}$ & ${B^{(0)}}^{(k,k),(k,k)}$ & ${B^{(0)}}^{(k,l),(l,l)}$ \\
    \hline
    \multirow{2}{*}{}& $\alpha^T$ & ${{\gamma^{(k)}}^{(k,k)}}^T$ & ${{\gamma^{(k,k)}}^{(k,k)}}^T$ \\
    \cmidrule{2-4}
    & & ${{\gamma^{(k)}}^{(k,l)}}^T$ & ${{\gamma^{(k,l)}}^{(k,l)}}^T$
\end{tabular}
\end{equation*}
as
\begin{equation*}
\renewcommand{\arraystretch}{1.3}
\begin{tabular}{r|ccc|ccc|ccc}
    & & & & & & \\
    $0$ & & & & & & \\
    $\frac{2}{3}$ & $\frac{2}{3}$ & &  & $1$ & &  & \\
    $\frac{2}{3}$ & $-\frac{1}{3}$ & $1$ & & $0$ & $0$ &  & & \\
    \cmidrule{1-10}
    $0$ & & & & & & \\
    $1$ & $1$ & &  & $1$ & &  & $1$ \\
    $1$ & $1$ & $0$ & & $-1$ & $0$ &  & $-1$ & $0$ \\
    \hline
    & $\frac{1}{4}$ & $\frac{1}{2}$ & $\frac{1}{4}$ &
    $\frac{1}{2}$ & $\frac{1}{4}$ & $\frac{1}{4}$ & $0$ & $\frac{1}{2}$ & $-\frac{1}{2}$\\
    \cmidrule{2-10}
    & & & & $-\frac{1}{2}$ & $\frac{1}{4}$ & $\frac{1}{4}$ & $0$ &
    $\frac{1}{2}$ & $-\frac{1}{2}$
\end{tabular}
\end{equation*}
\ \\
\noindent The parameters not shown in this array are 0.
The random variables are set to the following discrete ones.
For $k\in[m]$, $\theta_k$ takes $\pm\sqrt{3h}$, $0$, and $-\sqrt{3h}$ with probability $\frac{1}{6}$, $\frac{2}{3}$, and $\frac{1}{6}$, respectively.
For $k,l\in[m]$, $\theta_{(k,l)}=\frac{1}{2}\left(\theta_k\theta_l+V_{k,l}\right)$, where $V_{k,l}$ independent of $\{\theta_k\}_{k\in[m]}$ is defined as follows: if $l \le k-1$, $V_{k,l}$ takes $\pm h$ with probability $\frac{1}{2}$, if $l = k$, $V_{k,k}=-h$, and if $l \ge k+1$, $V_{k,l}=-V_{l,k}$.
RI1WM has a weak order 2 if the following conditions hold:
\begin{enumerate}
    \item $\|a(t,x)\|+\|b(t,x)\| \le C(1+\|x\|)$ and $\|a(t,x)-a(t,y)\|+\|b(t,x)-b(t,y)\| \le C(\|x-y\|)$ hold for any $t\in[0,T]$ and $x,y\in\mathbb{R}^d$ with some real constant $C$.
    \item all the components in $a$ and $b$, along with $f$, belong to $C^{6}_P(\mathbb{R}^d, \mathbb{R})$.
\end{enumerate}
RI1WM is applicable even if the commutativity of the noise,
\begin{align}
 (b^{k}(t,x)\cdot\nabla)b^{l}(t,x) -  (b^{l}(t,x)\cdot\nabla)b^{k}(t,x)=0,  
\end{align}
does not hold, where $\nabla=(\partial/\partial x_1,\ldots,\partial/\partial x_d)^T$, although some kinds of SRK methods require the commutativity of the noise.

\section{Proofs \label{app:proof}}

Before proving Theorem \ref{th:main}, we present a lemma used in the proof.

\begin{lemma}
    Consider Problem \ref{prob:main}.
    Also consider another SDE
    \begin{align}
    d\tilde{X}_t=\sum_{k=1}^K \tilde{\gamma}_k(t) \alpha_k(X_t) dt + \sum_{k=1}^K \tilde{\gamma}_k(t) \beta_k(X_t) dW_t, \tilde{X}_0=X_0
    \label{eq:MVSDEPert}
    \end{align}
    where the functions $\tilde{\gamma}_1,\ldots,\tilde{\gamma}_K:[0,T]\rightarrow\mathbb{R}$ satisfy
    \begin{align}
        \max_{t \in [0,T]} |\tilde{\gamma}_k(t)-\gamma_k(t)|\le\delta
        \label{eq:gammadelta}
    \end{align}
    with some common $\delta\in \mathbb{R}_+$, and assume that there exists a unique solution of Eq.~\eqref{eq:MVSDEPert}.
    Then, for any $t\in[0,T]$ and Lipschitz continuous function $f$ such that finite $\mathbb{E}[f(\tilde{X}_t)]$ and $\mathbb{E}[f(X_t)]$ exist,
    \begin{align}
        \left|\mathbb{E}[f(\tilde{X}_t)]-\mathbb{E}[f(X_t)]\right|
        \le  2\sqrt{(t+m)d}KLU\delta\exp\left(2(t+m)td^2K^2U^4\right)
    \end{align}
    holds, where $L$ is $f$'s Lipshitz constant. 
    \label{lem:SDEPert}
\end{lemma}

\begin{proof}
    We can write $\tilde{X}_t - X_t$ as
    \begin{align}
        \tilde{X}_t - X_t
        =& \int_0^t \sum_{k=1}^K(\tilde{\gamma}_k(s)-\gamma_k(s))\alpha_k(\tilde{X}_s)ds 
        +\int_0^t \sum_{k=1}^K\gamma_k(s)(\alpha_k(\tilde{X}_s)-\alpha_k(X_s))ds \nonumber \\
        & + \int_0^t \sum_{k=1}^K (\tilde{\gamma}_k(s)-\gamma_k(s))\beta_k(\tilde{X}_s)dW_s + \int_0^t \sum_{k=1}^K \gamma_k(s)(\beta_k(\tilde{X}_s)-\beta_k(X_s)) dW_s.
    \end{align}
    We thus have a bound on $\Delta_t\coloneqq \mathbb{E}[|\tilde{X}_t - X_t|^2]$ as
    \begin{align}
        \Delta_t &\le 4\left[\mathbb{E}\left[\left|\int_0^t \sum_{k=1}^K\left(\tilde{\gamma}_k(s)-\gamma_k(s)\right)\alpha_k(\tilde{X}_s)ds\right|^2\right] \right.\nonumber \\
        &\qquad +\mathbb{E}\left[\left|\int_0^t \sum_{k=1}^K\gamma_k(s)\left(\alpha_k(\tilde{X}_s)-\alpha_k(X_s)\right)ds\right|^2\right] \nonumber \\
        &\qquad  + \mathbb{E}\left[\left|\int_0^t \sum_{k=1}^K\left(\tilde{\gamma}_k(s)-\gamma_k(s)\right)\beta_k(\tilde{X}_s)dW_s\right|^2\right] \nonumber \\
        &\qquad  +\left.\mathbb{E}\left[\left|\int_0^t \sum_{k=1}^K\gamma_k(s)(\beta_k(\tilde{X}_s)-\beta_k(X_s))dW_s\right|^2\right]\right],
    \end{align}
    where we use an inequality $|w+x+y+z|^2 \le 4(|w|^2+|x|^2+|y|^2+|z|^2)$ holds for any $x,y,z,w\in\mathbb{R}^d$.
    Let us bound each term in the right hand side.
    For the first term, we have
    \begin{align}
        \mathbb{E}\left[\left|\int_0^t \sum_{k=1}^K\left(\tilde{\gamma}_k(s)-\gamma_k(s)\right)\alpha_k(\tilde{X}_s)ds\right|^2\right] \le & \mathbb{E}\left[\left(\int_0^t \sum_{k=1}^K \left|\tilde{\gamma}_k(s)-\gamma_k(s)\right|\cdot \left|\alpha_k(\tilde{X}_s)\right|ds\right)^2\right] \nonumber \\
        \le & d(\delta UKt)^2,
    \end{align}
    where Eq.~\eqref{eq:gammadelta} and Assumption~\ref{ass:alphabetaBound} are used.
    For the second term, we see that
    \begin{align}
        \left|\int_0^t \sum_{k=1}^K\gamma_k(s)\left(\alpha_k(\tilde{X}_s)-\alpha_k(X_s)\right)ds\right|
        \le & \int_0^t \sum_{k=1}^K\left|\gamma_k(s)\right|\cdot \left|\alpha_k(\tilde{X}_s)-\alpha_k(X_s)\right|ds \nonumber \\
        \le & dKU^2\int_0^t \left|\tilde{X}_s-X_s\right|ds,
        \label{eq:2ndTermTemp}
    \end{align}
    from Assumption~\ref{ass:alphabetaBound} and $\alpha_k$'s Lipschitz continuity with Lipschitz constant $dU$, which follows from Assumption~\ref{ass:alphabetaBound}.
    Eq.~\eqref{eq:2ndTermTemp} implies that
    \begin{align}
        \mathbb{E}\left[\left|\int_0^t \sum_{k=1}^K\gamma_k(s)\left(\alpha_k(\tilde{X}_s)-\alpha_k(X_s)\right)ds\right|^2\right]
        \le & d^2K^2U^4 \mathbb{E}\left[\left(\int_0^t \left|\tilde{X}_s-X_s\right|ds\right)^2\right] \nonumber \\
        \le & d^2K^2U^4t \int_0^t \Delta_s ds,
    \end{align}
    where we use $\left(\int_0^t \left|\tilde{X}_s-X_s\right|ds\right)^2\le t \int_0^t \left|\tilde{X}_s-X_s\right|^2ds$, a kind of the Cauchy–Schwarz inequality.
    The third term is bounded as
    \begin{align}
        \mathbb{E}\left[\left|\int_0^t \sum_{k=1}^K(\tilde{\gamma}_k(s)-\gamma_k(s))\beta_k(\tilde{X}_s)dW_s\right|^2 \right]
        \le & \mathbb{E}\left[\int_0^t \left|\sum_{k=1}^K(\tilde{\gamma}_k(s)-\gamma_k(s))\beta_k(\tilde{X}_s)\right|^2ds \right] \nonumber \\
        \le & \mathbb{E}\left[\int_0^t \left(\sum_{k=1}^K\left|\tilde{\gamma}_k(s)-\gamma_k(s)\right|\cdot\left|\beta_k(\tilde{X}_s)\right|\right)^2ds \right] \nonumber \\
        \le & K^2\delta^2 dmU^2t.
    \end{align}
    Lastly, the fourth term is bounded as
    \begin{align}
        \mathbb{E}\left[\left|\int_0^t \sum_{k=1}^K\gamma_k(s)(\beta_k(\tilde{X}_s)-\beta_k(X_s))dW_s\right|^2\right] 
        \le & \mathbb{E}\left[\int_0^t \left|\sum_{k=1}^K\gamma_k(s)(\beta_k(\tilde{X}_s)-\beta_k(X_s))\right|^2ds\right] \nonumber \\
        \le & d^2mK^2 U^4 \int_0^t \Delta_s ds,
    \end{align}
    similarly to the second term.
    Combining the above bounds, we have
    \begin{align}
        \Delta_t \le 4(t+m)dK^2U^2\left[\delta^2t+dU^2\int_0^t \Delta_s ds\right].
    \end{align}
    Applying Gr\"onwall's inequality, we get
    \begin{align}
        \Delta_t \le 4(t+m)dK^2U^2\delta^2\exp\left(4(t+m)td^2K^2U^4\right).
    \end{align}
    This implies that
    \begin{align}
        \mathbb{E}[|\tilde{X}_t-X_t|] \le \sqrt{\Delta_t}\le 2\sqrt{(t+m)d}KU\delta\exp\left(2(t+m)td^2K^2U^4\right),
    \end{align}
    and for the Lipschitz continuous function $f$,
    \begin{align}
        \left|\mathbb{E}[f(\tilde{X}_t)]-\mathbb{E}[f(X_t)]\right|
        \le & L\mathbb{E}[|\tilde{X}_t-X_t|] \nonumber \\
        \le & 2\sqrt{(t+m)d}KLU\delta\exp\left(2(t+m)td^2K^2U^4\right).
    \end{align}

\end{proof}

\begin{proof}[Proof of Theorem \ref{th:main}]

Let us use induction.
Assume that for $i\in[n_{\rm t}]$, $\tilde{\gamma}_{k,i}$ in Eq.~\eqref{eq:gammatil} satisfies
\begin{align}
\max_{t \in [0,t_{i+1}]} |\tilde{\gamma}_{k,i}(t)-\gamma_k(t)|\le\epsilon,
\label{eq:gammaErr}
\end{align}
which holds for $i=0$ since
\begin{align}
    |\gamma_{k,0}-\gamma_k(t)|\le \max_{s\in[0,t_1]} |\gamma^\prime(s)| t \le Uh_{\rm I} \le \epsilon
\end{align}
holds by the mean value theorem.

We define a random process $X^{(i)}_t$ by the SDE
\begin{align}
dX^{(i)}_t=a_i(t,X^{(i)}_t) dt + b_i(t,X^{(i)}_t) dW_t
\label{eq:SDEApprox}
\end{align}
with $a_i$ and $b_i$ defined in Algorithm~\ref{alg:main}.
$\tilde{X}_i$ is an approximation of $X^{(i)}_{t_i}$ by the discretization method $F$.

We first consider the error of $\tilde{X}_i$ as an approximation of $X^{(i)}_{t_i}$ in terms of the expectation of $\varphi_k$ in a way similar to the proof of Theorem 9.1 in \cite{milstein2013numerical}.
We denote by $X^{(i),s,x}_t$ the solution of Eq.~\eqref{eq:SDEApprox} starting from the time $s$ with an initial value $x$.
We also define
\begin{align}
    u_{i+1,k}(t,x)\coloneqq \mathbb{E}[\varphi_k(X^{(i),t,x}_{t_{i+1}})]
\end{align}
for $t\in[0,t_{i+1}]$ and $x\in\mathbb{R}^d$.
Then, noting that
\begin{align}
    X^{(i),t_{j},\tilde{X}_j}_{t_{i+1}}=X^{(i),t_{j+1},X^{(i),t_{j},\tilde{X}_j}_{t_{j+1}}}_{t_{i+1}}
\end{align}
holds for any $j \le i-1$, we have
\begin{align}
    &\mathbb{E}\left[\varphi_k\left(X^{(i)}_{t_{i+1}}\right)\right]-\mathbb{E}\left[\varphi_k(\tilde{X}_{i+1})\right] \nonumber \\
    =&\mathbb{E}\left[\varphi_k\left(X^{(i),t_1,X^{(i),t_0,\tilde{X}_0}_{t_1}}_{t_{i+1}}\right)\right] + \sum_{j=1}^{i} \left(\mathbb{E}\left[\varphi_k\left(X^{(i),t_{j+1},X^{(i),t_{j},\tilde{X}_j}_{t_{j+1}}}_{t_{i+1}}\right)\right] -\mathbb{E}\left[\varphi_k\left(X^{(i),t_{j},\tilde{X}_{j}}_{t_{i+1}}\right)\right]\right) \nonumber \\
    &-\mathbb{E}\left[\varphi(\tilde{X}_{i+1})\right]  \nonumber \\
    =&\sum_{j=0}^{i-1} \left(\mathbb{E}\left[\varphi_k\left(X^{(i),t_{j+1},X^{(i),t_{j},\tilde{X}_j}_{t_{j+1}}}_{t_{i+1}}\right)\right]-\mathbb{E}\left[\varphi_k\left(X^{(i),t_{j+1},\tilde{X}_{j+1}}_{t_{i+1}}\right)\right]\right)+ \mathbb{E}\left[\varphi_k\left(X^{(i),t_{i},\tilde{X}_i}_{t_{i+1}}\right)\right] \nonumber \\
    &-\mathbb{E}\left[\varphi(\tilde{X}_{i+1})\right] \nonumber \\
    =&\sum_{j=0}^{i-1} \left(\mathbb{E}\left[\mathbb{E}\left[\varphi_k\left(X^{(i),t_{j+1},X^{(i),t_{j},\tilde{X}_j}_{t_{j+1}}}_{t_{i+1}}\right)~\middle|~X^{(i),t_{j},\tilde{X}_j}_{t_{j+1}}\right]\right]-\mathbb{E}\left[\mathbb{E}\left[\varphi_k\left(X^{(i),t_{j+1},\tilde{X}_{j+1}}_{t_{i+1}}\right)~\middle|~\tilde{X}_{j+1}\right]\right]\right) \nonumber \\
    &+ \mathbb{E}\left[\varphi_k\left(X^{(i),t_{i},\tilde{X}_i}_{t_{i+1}}\right)\right]-\mathbb{E}\left[\varphi(\tilde{X}_{i+1})\right] \nonumber \\
    =& \sum_{j=0}^{i-1} \left(\mathbb{E}\left[u_k\left(t_{j+1},X^{(i),t_{j},\tilde{X}_j}_{t_{j+1}}\right)\right]-\mathbb{E}\left[u_k\left(t_{j+1},\tilde{X}_{j+1}\right)\right]\right)+ \mathbb{E}\left[\varphi_k\left(X^{(i),t_{i},\tilde{X}_i}_{t_{i+1}}\right)\right]-\mathbb{E}\left[\varphi(\tilde{X}_{i+1})\right] \nonumber \\
    =& \sum_{j=0}^{i-1} \mathbb{E}\left[\mathbb{E}\left[u_k\left(t_{j+1},X^{(i),t_{j},\tilde{X}_j}_{t_{j+1}}\right)~\middle|~\tilde{X}_j\right]-\mathbb{E}\left[u_k\left(t_{j+1},\tilde{X}_{j+1}\right)~\middle|~\tilde{X}_j\right]\right]\nonumber \\
    &+\mathbb{E}\left[\mathbb{E}\left[\varphi_k\left(X^{(i),t_{i},\tilde{X}_i}_{t_{i+1}}\right)~\middle|~\tilde{X}_i\right]-\mathbb{E}\left[\varphi(\tilde{X}_{i+1})~\middle|~\tilde{X}_i\right]\right].
    \label{eq:weakErrOneStep}
\end{align}
According to \cite{milstein2013numerical}, under Assumption \ref{ass:alphabetaBound} and \ref{ass:basis}, $u_k(t,x)$ is in $\mathbb{C}^{2(\lceil p \rceil+1)}_P$ with respect to $x$.
Because of Assumption \ref{ass:basis}, $\varphi_k$ is also in $\mathbb{C}^{2(\lceil p \rceil+1)}_P$.
Then, since the method $F$ has weak order $p$,
\begin{align}
    &\left|\mathbb{E}\left[u_k\left(t_{j+1},X^{(i),t_{j},\tilde{X}_j}_{t_{j+1}}\right)~\middle|~\tilde{X}_j\right]-\mathbb{E}\left[u_k\left(t_{j+1},\tilde{X}_{j+1}\right)~\middle|~\tilde{X}_j\right]\right| \nonumber \\
    & \qquad\qquad\qquad\qquad\qquad\qquad\qquad\quad \le \kappa(1+|\tilde{X}_j|^{2r})h_j^{p+1}, \nonumber \\
    & \left|\mathbb{E}\left[\varphi_k\left(X^{(i),t_{i},\tilde{X}_i}_{t_{i+1}}\right)~\middle|~\tilde{X}_i\right]-\mathbb{E}\left[\varphi(\tilde{X}_{i+1})~\middle|~\tilde{X}_i\right]\right|\le \kappa(1+|\tilde{X}_i|^{2r})h_i^{p+1}
\end{align}
holds for the random processes $X^{(i),t_{j},\tilde{X}_j}_t,j=0,\ldots,i$ and their weak approximations $\tilde{X}_{j+1}$, with some $\kappa\in\mathbb{R}_+$ and $r\in\mathbb{N}$.
Using these along with the assumption $\mathbb{E}[|\tilde{X}_i|^{2r}]<\infty$ in Eq.~\eqref{eq:weakErrOneStep} yields
\begin{align}
    \left|\mathbb{E}\left[\varphi_k\left(X^{(i)}_{t_{i+1}}\right)\right]-\mathbb{E}\left[\varphi_k(\tilde{X}_{i+1})\right]\right|\le \kappa^\prime \sum_{j=0}^i h_j^{p+1}
\end{align}
with some $\kappa^\prime\in\mathbb{R}_+$, and by plugging Eq.~\eqref{eq:hIandII} into this and some algebra, we finally get
\begin{align}
    \left|\mathbb{E}\left[\varphi_k\left(X^{(i)}_{t_{i+1}}\right)\right]-\mathbb{E}\left[\varphi_k(\tilde{X}_{i+1})\right]\right|\le \frac{\epsilon}{12}.
    \label{eq:EphiErr1}
\end{align}

Suppose that all the runs of QMCI in Algorithm~\ref{alg:main} succeed.
We then have $\hat{\gamma}_{k,i}$ such that
\begin{equation}
    \left|\hat{\gamma}_{k,i+1}-\mathbb{E}\left[\varphi_k(\tilde{X}_{i+1})\right]\right|\le \frac{\epsilon}{12}.
    \label{eq:EphiErr2}
\end{equation}
Besides, under the induction assumption as Eq.~\eqref{eq:gammaErr} and the assumption that Eq.~\eqref{eq:ShortTimeAssum} holds, Lemma~\ref{lem:SDEPert} implies
\begin{align}
    \left|\gamma_{k}(t_{i+1})-\mathbb{E}\left[\varphi_k\left(X^{(i)}_{t_{i+1}}\right)\right]\right|\le \frac{\epsilon}{12}
    \label{eq:EphiErr3}
\end{align}
(note that for $\varphi_k$, the Lipschitz constant $L$ satisfies $L \le \sqrt{d}U$ under Assumption~\ref{ass:basis}).
Combining Eqs.~\eqref{eq:EphiErr1}, \eqref{eq:EphiErr2}, and \eqref{eq:EphiErr3} yields
\begin{align}
    \left|\hat{\gamma}_{k,i+1}-\gamma_{k}(t_{i+1})\right|\le  \frac{\epsilon}{4}.
    \label{eq:gammahatErr}
\end{align}

For $i\le n^{\rm I}_t-2$, we see that
\begin{align}
    |\tilde{\gamma}_{k,i+1}(t)-\gamma_k(t)| 
    \le & |\hat{\gamma}_{k,i+1}-\gamma_k(t_{i+1})| + |\gamma_k(t)-\gamma_k(t_{i+1})|  \nonumber \\
    \le & \frac{\epsilon}{4} + \max_{s\in[t_i,t_{i+1}]} |\gamma^\prime_k(s)| h_{\rm I}  \nonumber \\
    \le & \epsilon
\end{align}
holds for any $t\in[t_{i+1},t_{i+2}]$.
Here, Eq.~\eqref{eq:gammahatErr} and the mean value theorem are used at the second inequality, and Assumption~\ref{ass:gammaBound} and Eq.~\eqref{eq:hIandII} are used at the last inequality.
For $i\ge n^{\rm I}_t-1$, we get
\begin{align}
    \left|\gamma_{k}^\prime(t_{i+1})-\hat{\gamma}_{k,i+1}^\prime\right|
    = & \left|\gamma_{k}^\prime(t_{i+1})-\frac{\hat{\gamma}_{k,i+1}-\hat{\gamma}_{k,i^\prime}}{h_{\rm II}}\right| \nonumber \\
    \le & \left|\gamma_{k}^\prime(t_{i+1})-\frac{\gamma_k(t_{i+1})-\gamma_k(t_{i^\prime})}{h_{\rm II}}\right|+\frac{\left|\hat{\gamma}_{k,i+1}-\gamma_k(t_{i+1})\right|}{h_{\rm II}}+\frac{\left|\hat{\gamma}_{k,i^\prime}-\gamma_k(t_{i^\prime})\right|}{h_{\rm II}} \nonumber \\
    \le & \frac{1}{2}Uh_{\rm II}+\frac{\epsilon}{2h_{\rm II}},
    \label{eq:gammahatprErr}
\end{align}
where $i^\prime=i$ if $i\ge n^{\rm I}_t$ and $i^\prime=0$ if $i= n^{\rm I}_t-1$, using Taylor's Theorem, Assumption~\ref{ass:gammaBound}, and Eq.~\eqref{eq:gammahatErr} at the last inequality.
We thus have
\begin{align}
    |\tilde{\gamma}_{k,i+1}(t)-\gamma_k(t)| 
    =&|\hat{\gamma}_{k,i+1}^\prime(t-t_{i+1})+\hat{\gamma}_{k,i+1}-\gamma_k(t)| \nonumber \\
    \le & \left| \gamma_{k}^\prime(t_{i+1})(t-t_{i+1})+\gamma_k(t_{i+1})-\gamma_k(t) \right| \nonumber \\
    & \quad +\left|\left(\gamma_{k}^\prime(t_{i+1})-\hat{\gamma}_{k,i+1}^\prime\right)(t-t_{i+1})\right|+\left|\hat{\gamma}_{k,i+1}-\gamma_k(t_{i+1})\right| \nonumber \\
    \le & \frac{1}{2}Uh_{i+1}^2 + \left(\frac{1}{2}Uh_{\rm II}+\frac{\epsilon}{2h_{\rm II}}\right)h_{i+1}+\frac{\epsilon}{4} \nonumber \\
    \le & \epsilon
\end{align}
for $t\in[t_{i+1},t_{i+2}]$, where Taylor's theorem, Assumption~\ref{ass:gammaBound}, Eq.~\eqref{eq:gammahatErr} and Eq.~\eqref{eq:gammahatprErr} are used at the second inequality, and Eq.~\eqref{eq:hIandII} is used at the last inequality along with some algebra.
In summary, $|\tilde{\gamma}_{k,i+1}(t)-\gamma_k(t)|\le\epsilon$ holds for any $k\in[K]$ and $t\in[t_{i+1},t_{i+2}]$ whether $i\le n^{\rm I}_t-2$ or $i\ge n^{\rm I}_t-1$.
That is, the induction assumption also holds in the next step: Eq.~\eqref{eq:gammaErr} also holds when $i$ is replaced with $i+1$.

The induction implies that Eq.~\eqref{eq:gammahatErr} also holds for any $i=0,\ldots,n_{\rm t}-1$.
In particular, $\left|\hat{\gamma}_{k,n_{\rm t}}-\mathbb{E}\left[\varphi_k\left(X_T\right)\right]\right|\le \frac{\epsilon}{4}$ holds.
We also see that
\begin{align}
    \left|\hat{E}-\mathbb{E}\left[\phi\left(X_T\right)\right]\right|\le \frac{\epsilon}{4} < \epsilon
\end{align}
holds, since $\phi$ is also in $\mathbb{C}^{2(\lceil p \rceil+1)}_P$.
This concludes the proof on the accuracy.

Next, we evaluate the complexity of Algorithm~\ref{alg:main}.
Because of Theorem~\ref{th:QMCI}, QMCI in step 9 (resp. 11) in the algorithm makes
\begin{align}
    O\left(\frac{U}{\epsilon_{\rm QMCI}}\log^{3/2}\left(\frac{U}{\epsilon_{\rm QMCI}}\right)\log\log\left(\frac{U}{\epsilon_{\rm QMCI}}\right)\log\left(\frac{1}{\eta^\prime}\right)\right)
\end{align}
queries to $U_{\tilde{X}_{i+1}}$ and $U_{\varphi_k}$ (resp. $U_\phi$).
$U_{\tilde{X}_{i+1}}$ contains $i+1$ queries to $U_{F^h_{a,b,t}}$, and $i+1$ is at most $n_{\rm t}=n_{\rm t}^{\rm I} + n_{\rm t}^{\rm I}=\frac{h_{\rm II}}{h_{\rm I}} + \frac{T}{h_{\rm II}} -1$, which is $O\left(1/\epsilon^{1/p}\right)$ in the current setting.
Combining these and Eqs.~\eqref{eq:etaPr} and \eqref{eq:epsQMCI}, and leaving only the dependencies on $\epsilon$ and $\eta$, we obtain the query number bounds in Eq.~\eqref{eq:query1} and \eqref{eq:query2}.

Lastly, let us see the success probability of Algorithm~\ref{alg:main}.
The success probability of each QMCI in the algorithm is at least $1-\eta^\prime$.
The number of the QMCIs in the algorithm is $K(n_{\rm t}-1)+1$.
Thus, the probability that all of the QMCIs succeed, which means that Algorithm~\ref{alg:main} successfully outputs an $\epsilon$-approximation of $\mathbb{E}[\phi(X_T)]$, is at least
\begin{align}
(1-\eta^\prime)^{K(n_{\rm t}-1)+1} \ge 1-(K(n_{\rm t}-1)+1)\eta^\prime =  1-\eta.    
\end{align}

\end{proof}

\end{appendices}

\bibliography{reference}


\begin{thebibliography}{43}
\ifx \bisbn   \undefined \def \bisbn  #1{ISBN #1}\fi
\ifx \binits  \undefined \def \binits#1{#1}\fi
\ifx \bauthor  \undefined \def \bauthor#1{#1}\fi
\ifx \batitle  \undefined \def \batitle#1{#1}\fi
\ifx \bjtitle  \undefined \def \bjtitle#1{#1}\fi
\ifx \bvolume  \undefined \def \bvolume#1{\textbf{#1}}\fi
\ifx \byear  \undefined \def \byear#1{#1}\fi
\ifx \bissue  \undefined \def \bissue#1{#1}\fi
\ifx \bfpage  \undefined \def \bfpage#1{#1}\fi
\ifx \blpage  \undefined \def \blpage #1{#1}\fi
\ifx \burl  \undefined \def \burl#1{\textsf{#1}}\fi
\ifx \doiurl  \undefined \def \doiurl#1{\url{https://doi.org/#1}}\fi
\ifx \betal  \undefined \def \betal{\textit{et al.}}\fi
\ifx \binstitute  \undefined \def \binstitute#1{#1}\fi
\ifx \binstitutionaled  \undefined \def \binstitutionaled#1{#1}\fi
\ifx \bctitle  \undefined \def \bctitle#1{#1}\fi
\ifx \beditor  \undefined \def \beditor#1{#1}\fi
\ifx \bpublisher  \undefined \def \bpublisher#1{#1}\fi
\ifx \bbtitle  \undefined \def \bbtitle#1{#1}\fi
\ifx \bedition  \undefined \def \bedition#1{#1}\fi
\ifx \bseriesno  \undefined \def \bseriesno#1{#1}\fi
\ifx \blocation  \undefined \def \blocation#1{#1}\fi
\ifx \bsertitle  \undefined \def \bsertitle#1{#1}\fi
\ifx \bsnm \undefined \def \bsnm#1{#1}\fi
\ifx \bsuffix \undefined \def \bsuffix#1{#1}\fi
\ifx \bparticle \undefined \def \bparticle#1{#1}\fi
\ifx \barticle \undefined \def \barticle#1{#1}\fi
\bibcommenthead
\ifx \bconfdate \undefined \def \bconfdate #1{#1}\fi
\ifx \botherref \undefined \def \botherref #1{#1}\fi
\ifx \url \undefined \def \url#1{\textsf{#1}}\fi
\ifx \bchapter \undefined \def \bchapter#1{#1}\fi
\ifx \bbook \undefined \def \bbook#1{#1}\fi
\ifx \bcomment \undefined \def \bcomment#1{#1}\fi
\ifx \oauthor \undefined \def \oauthor#1{#1}\fi
\ifx \citeauthoryear \undefined \def \citeauthoryear#1{#1}\fi
\ifx \endbibitem  \undefined \def \endbibitem {}\fi
\ifx \bconflocation  \undefined \def \bconflocation#1{#1}\fi
\ifx \arxivurl  \undefined \def \arxivurl#1{\textsf{#1}}\fi
\csname PreBibitemsHook\endcsname

\bibitem[\protect\citeauthoryear{Montanaro}{2015}]{montanaro2015}
\begin{barticle}
\bauthor{\bsnm{Montanaro}, \binits{A.}}:
\batitle{Quantum speedup of monte carlo methods}.
\bjtitle{Proc. R. Soc. A}
\bvolume{471}(\bissue{2181}),
\bfpage{20150301}
(\byear{2015})
\doiurl{10.1098/rspa.2015.0301}
\end{barticle}
\endbibitem

\bibitem[\protect\citeauthoryear{Rebentrost et~al.}{2018}]{Rebentrost2018}
\begin{barticle}
\bauthor{\bsnm{Rebentrost}, \binits{P.}},
\bauthor{\bsnm{Gupt}, \binits{B.}},
\bauthor{\bsnm{Bromley}, \binits{T.R.}}:
\batitle{Quantum computational finance: Monte carlo pricing of financial derivatives}.
\bjtitle{Phys. Rev. A}
\bvolume{98},
\bfpage{022321}
(\byear{2018})
\doiurl{10.1103/PhysRevA.98.022321}
\end{barticle}
\endbibitem

\bibitem[\protect\citeauthoryear{Stamatopoulos et~al.}{2020}]{Stamatopoulos2020optionpricingusing}
\begin{barticle}
\bauthor{\bsnm{Stamatopoulos}, \binits{N.}},
\bauthor{\bsnm{Egger}, \binits{D.J.}},
\bauthor{\bsnm{Sun}, \binits{Y.}},
\bauthor{\bsnm{Zoufal}, \binits{C.}},
\bauthor{\bsnm{Iten}, \binits{R.}},
\bauthor{\bsnm{Shen}, \binits{N.}},
\bauthor{\bsnm{Woerner}, \binits{S.}}:
\batitle{Option {P}ricing using {Q}uantum {C}omputers}.
\bjtitle{{Quantum}}
\bvolume{4},
\bfpage{291}
(\byear{2020})
\doiurl{10.22331/q-2020-07-06-291}
\end{barticle}
\endbibitem

\bibitem[\protect\citeauthoryear{Chakrabarti et~al.}{2021}]{Chakrabarti2021thresholdquantum}
\begin{barticle}
\bauthor{\bsnm{Chakrabarti}, \binits{S.}},
\bauthor{\bsnm{Krishnakumar}, \binits{R.}},
\bauthor{\bsnm{Mazzola}, \binits{G.}},
\bauthor{\bsnm{Stamatopoulos}, \binits{N.}},
\bauthor{\bsnm{Woerner}, \binits{S.}},
\bauthor{\bsnm{Zeng}, \binits{W.J.}}:
\batitle{A {T}hreshold for {Q}uantum {A}dvantage in {D}erivative {P}ricing}.
\bjtitle{{Quantum}}
\bvolume{5},
\bfpage{463}
(\byear{2021})
\doiurl{10.22331/q-2021-06-01-463}
\end{barticle}
\endbibitem

\bibitem[\protect\citeauthoryear{Miyamoto}{2022}]{miyamoto2022bermudan}
\begin{barticle}
\bauthor{\bsnm{Miyamoto}, \binits{K.}}:
\batitle{Bermudan option pricing by quantum amplitude estimation and chebyshev interpolation}.
\bjtitle{{EPJ Quantum Technol.}}
\bvolume{9},
\bfpage{3}
(\byear{2022})
\doiurl{10.1140/epjqt/s40507-022-00124-3}
\end{barticle}
\endbibitem

\bibitem[\protect\citeauthoryear{Kaneko et~al.}{2022}]{kaneko2022quantum}
\begin{barticle}
\bauthor{\bsnm{Kaneko}, \binits{K.}},
\bauthor{\bsnm{Miyamoto}, \binits{K.}},
\bauthor{\bsnm{Takeda}, \binits{N.}},
\bauthor{\bsnm{Yoshino}, \binits{K.}}:
\batitle{Quantum pricing with a smile: Implementation of local volatility model on quantum computer}.
\bjtitle{{EPJ Quantum Technol.}}
\bvolume{9},
\bfpage{7}
(\byear{2022})
\doiurl{10.1140/epjqt/s40507-022-00125-2}
\end{barticle}
\endbibitem

\bibitem[\protect\citeauthoryear{Doriguello et~al.}{2022}]{doriguello2022}
\begin{bchapter}
\bauthor{\bsnm{Doriguello}, \binits{J.F.}},
\bauthor{\bsnm{Luongo}, \binits{A.}},
\bauthor{\bsnm{Bao}, \binits{J.}},
\bauthor{\bsnm{Rebentrost}, \binits{P.}},
\bauthor{\bsnm{Santha}, \binits{M.}}:
\bctitle{{Quantum Algorithm for Stochastic Optimal Stopping Problems with Applications in Finance}}.
In: \bbtitle{17th Conference on the Theory of Quantum Computation, Communication and Cryptography (TQC 2022)},
pp. \bfpage{2}--\blpage{1224}
(\byear{2022}).
\doiurl{10.4230/LIPIcs.TQC.2022.2}
\end{bchapter}
\endbibitem

\bibitem[\protect\citeauthoryear{Wang et~al.}{2021}]{pmlr-v139-wang21w}
\begin{bchapter}
\bauthor{\bsnm{Wang}, \binits{D.}},
\bauthor{\bsnm{Sundaram}, \binits{A.}},
\bauthor{\bsnm{Kothari}, \binits{R.}},
\bauthor{\bsnm{Kapoor}, \binits{A.}},
\bauthor{\bsnm{Roetteler}, \binits{M.}}:
\bctitle{Quantum algorithms for reinforcement learning with a generative model}.
In: \beditor{\bsnm{Meila}, \binits{M.}},
\beditor{\bsnm{Zhang}, \binits{T.}} (eds.)
\bbtitle{Proceedings of the 38th International Conference on Machine Learning}.
\bsertitle{Proceedings of Machine Learning Research},
vol. \bseriesno{139},
pp. \bfpage{10916}--\blpage{10926}.
\bpublisher{PMLR},
\blocation{Virtual}
(\byear{2021})
\end{bchapter}
\endbibitem

\bibitem[\protect\citeauthoryear{Wiedemann et~al.}{2023}]{wiedemann2023quantum}
\begin{botherref}
\oauthor{\bsnm{Wiedemann}, \binits{S.}},
\oauthor{\bsnm{Hein}, \binits{D.}},
\oauthor{\bsnm{Udluft}, \binits{S.}},
\oauthor{\bsnm{Mendl}, \binits{C.B.}}:
Quantum policy iteration via amplitude estimation and grover search {\textendash} towards quantum advantage for reinforcement learning.
Transactions on Machine Learning Research
(2023)
\end{botherref}
\endbibitem

\bibitem[\protect\citeauthoryear{Wan et~al.}{2023}]{Wan_Zhang_Li_Zhang_Sun_2023}
\begin{barticle}
\bauthor{\bsnm{Wan}, \binits{Z.}},
\bauthor{\bsnm{Zhang}, \binits{Z.}},
\bauthor{\bsnm{Li}, \binits{T.}},
\bauthor{\bsnm{Zhang}, \binits{J.}},
\bauthor{\bsnm{Sun}, \binits{X.}}:
\batitle{Quantum multi-armed bandits and stochastic linear bandits enjoy logarithmic regrets}.
\bjtitle{Proceedings of the AAAI Conference on Artificial Intelligence}
\bvolume{37}(\bissue{8}),
\bfpage{10087}--\blpage{10094}
(\byear{2023})
\doiurl{10.1609/aaai.v37i8.26202}
\end{barticle}
\endbibitem

\bibitem[\protect\citeauthoryear{Dai et~al.}{2023}]{NEURIPS2023_401aa72e}
\begin{bchapter}
\bauthor{\bsnm{Dai}, \binits{Z.}},
\bauthor{\bsnm{Lau}, \binits{G.K.R.}},
\bauthor{\bsnm{Verma}, \binits{A.}},
\bauthor{\bsnm{SHU}, \binits{Y.}},
\bauthor{\bsnm{Low}, \binits{B.K.H.}},
\bauthor{\bsnm{Jaillet}, \binits{P.}}:
\bctitle{Quantum bayesian optimization}.
In: \beditor{\bsnm{Oh}, \binits{A.}},
\beditor{\bsnm{Naumann}, \binits{T.}},
\beditor{\bsnm{Globerson}, \binits{A.}},
\beditor{\bsnm{Saenko}, \binits{K.}},
\beditor{\bsnm{Hardt}, \binits{M.}},
\beditor{\bsnm{Levine}, \binits{S.}} (eds.)
\bbtitle{Advances in Neural Information Processing Systems},
vol. \bseriesno{36},
pp. \bfpage{20179}--\blpage{20207}.
\bpublisher{Curran Associates, Inc.},
\blocation{New Orleans}
(\byear{2023})
\end{bchapter}
\endbibitem

\bibitem[\protect\citeauthoryear{Hikima et~al.}{2024}]{hikima2024quantum}
\begin{bchapter}
\bauthor{\bsnm{Hikima}, \binits{Y.}},
\bauthor{\bsnm{Murao}, \binits{K.}},
\bauthor{\bsnm{Takemori}, \binits{S.}},
\bauthor{\bsnm{Umeda}, \binits{Y.}}:
\bctitle{Quantum kernelized bandits}.
In: \beditor{\bsnm{Kiyavash}, \binits{N.}},
\beditor{\bsnm{Mooij}, \binits{J.M.}} (eds.)
\bbtitle{Proceedings of the Fortieth Conference on Uncertainty in Artificial Intelligence}.
\bsertitle{Proceedings of Machine Learning Research},
vol. \bseriesno{244},
pp. \bfpage{1640}--\blpage{1657}.
\bpublisher{PMLR},
\blocation{Barcelona}
(\byear{2024})
\end{bchapter}
\endbibitem

\bibitem[\protect\citeauthoryear{McKean}{1966}]{McKean1966}
\begin{barticle}
\bauthor{\bsnm{McKean}, \binits{H.P.}}:
\batitle{A class of markov processes associated with nonlinear parabolic equations}.
\bjtitle{Proceedings of the National Academy of Sciences}
\bvolume{56}(\bissue{6}),
\bfpage{1907}--\blpage{1911}
(\byear{1966})
\doiurl{10.1073/pnas.56.6.1907}
\end{barticle}
\endbibitem

\bibitem[\protect\citeauthoryear{Guyon and Henry-Labordere}{2013}]{guyon2013nonlinear}
\begin{bbook}
\bauthor{\bsnm{Guyon}, \binits{J.}},
\bauthor{\bsnm{Henry-Labordere}, \binits{P.}}:
\bbtitle{Nonlinear Option Pricing}.
\bpublisher{CRC Press},
\blocation{Boca Raton, FL}
(\byear{2013})
\end{bbook}
\endbibitem

\bibitem[\protect\citeauthoryear{Bossy and Talay}{1997}]{Bossy1997}
\begin{barticle}
\bauthor{\bsnm{Bossy}, \binits{M.}},
\bauthor{\bsnm{Talay}, \binits{D.}}:
\batitle{A stochastic particle method for the mckean-vlasov and the burgers equation}.
\bjtitle{Mathematics of Computation}
\bvolume{66}(\bissue{217}),
\bfpage{157}--\blpage{192}
(\byear{1997})
\end{barticle}
\endbibitem

\bibitem[\protect\citeauthoryear{M{\'e}l{\'e}ard}{1996}]{Meleard1996}
\begin{bbook}
\bauthor{\bsnm{M{\'e}l{\'e}ard}, \binits{S.}}:
In: \beditor{\bsnm{Talay}, \binits{D.}},
\beditor{\bsnm{Tubaro}, \binits{L.}} (eds.)
\bbtitle{Asymptotic behaviour of some interacting particle systems; McKean-Vlasov and Boltzmann models},
pp. \bfpage{42}--\blpage{95}.
\bpublisher{Springer},
\blocation{Berlin, Heidelberg}
(\byear{1996}).
\doiurl{10.1007/BFb0093177} .
\burl{https://doi.org/10.1007/BFb0093177}
\end{bbook}
\endbibitem

\bibitem[\protect\citeauthoryear{Belomestny and Schoenmakers}{2018}]{Belomestny2018}
\begin{barticle}
\bauthor{\bsnm{Belomestny}, \binits{D.}},
\bauthor{\bsnm{Schoenmakers}, \binits{J.}}:
\batitle{Projected particle methods for solving mckean--vlasov stochastic differential equations}.
\bjtitle{SIAM Journal on Numerical Analysis}
\bvolume{56}(\bissue{6}),
\bfpage{3169}--\blpage{3195}
(\byear{2018})
\doiurl{10.1137/17M1111024}
\end{barticle}
\endbibitem

\bibitem[\protect\citeauthoryear{Bossy and Talay}{1996}]{bossy1996convergence}
\begin{barticle}
\bauthor{\bsnm{Bossy}, \binits{M.}},
\bauthor{\bsnm{Talay}, \binits{D.}}:
\batitle{Convergence rate for the approximation of the limit law of weakly interacting particles: application to the burgers equation}.
\bjtitle{The Annals of Applied Probability}
\bvolume{6}(\bissue{3}),
\bfpage{818}--\blpage{861}
(\byear{1996})
\end{barticle}
\endbibitem

\bibitem[\protect\citeauthoryear{Antonelli and Kohatsu-Higa}{2002}]{antonelli2002rate}
\begin{barticle}
\bauthor{\bsnm{Antonelli}, \binits{F.}},
\bauthor{\bsnm{Kohatsu-Higa}, \binits{A.}}:
\batitle{Rate of convergence of a particle method to the solution of the mckean--vlasov equation}.
\bjtitle{The Annals of Applied Probability}
\bvolume{12}(\bissue{2}),
\bfpage{423}--\blpage{476}
(\byear{2002})
\end{barticle}
\endbibitem

\bibitem[\protect\citeauthoryear{Miyamoto}{2025}]{miyamoto2025}
\begin{barticle}
\bauthor{\bsnm{Miyamoto}, \binits{K.}}:
\batitle{Dividing quantum circuits for the time evolution of stochastic processes by orthogonal-series-density estimation}.
\bjtitle{Phys. Rev. A}
\bvolume{111},
\bfpage{042431}
(\byear{2025})
\doiurl{10.1103/PhysRevA.111.042431}
\end{barticle}
\endbibitem

\bibitem[\protect\citeauthoryear{R\"{o}\ss{}ler}{2006}]{Rossler01032006}
\begin{barticle}
\bauthor{\bsnm{R\"{o}\ss{}ler}, \binits{A.}}:
\batitle{Rooted tree analysis for order conditions of stochastic runge-kutta methods for the weak approximation of stochastic differential equations}.
\bjtitle{Stochastic Analysis and Applications}
\bvolume{24}(\bissue{1}),
\bfpage{97}--\blpage{134}
(\byear{2006})
\doiurl{10.1080/07362990500397699}
\end{barticle}
\endbibitem

\bibitem[\protect\citeauthoryear{R\"{o}\ss{}ler}{2010}]{Rossler2010}
\begin{barticle}
\bauthor{\bsnm{R\"{o}\ss{}ler}, \binits{A.}}:
\batitle{Runge–kutta methods for the strong approximation of solutions of stochastic differential equations}.
\bjtitle{SIAM Journal on Numerical Analysis}
\bvolume{48}(\bissue{3}),
\bfpage{922}--\blpage{952}
(\byear{2010})
\doiurl{10.1137/09076636X}
\end{barticle}
\endbibitem

\bibitem[\protect\citeauthoryear{Maruyama}{1955}]{Maruyama1955}
\begin{barticle}
\bauthor{\bsnm{Maruyama}, \binits{G.}}:
\batitle{Continuous markov processes and stochastic equations}.
\bjtitle{Rend. Circ. Mat. Palermo}
\bvolume{4},
\bfpage{48}--\blpage{90}
(\byear{1955})
\doiurl{10.1007/BF02846028}
\end{barticle}
\endbibitem

\bibitem[\protect\citeauthoryear{Milstein}{1995}]{milstein2013numerical}
\begin{bbook}
\bauthor{\bsnm{Milstein}, \binits{G.N.}}:
\bbtitle{Numerical Integration of Stochastic Differential Equations}.
\bpublisher{Springer},
\blocation{Dordrecht}
(\byear{1995})
\end{bbook}
\endbibitem

\bibitem[\protect\citeauthoryear{Li}{2005}]{LI200529}
\begin{barticle}
\bauthor{\bsnm{Li}, \binits{J.}}:
\batitle{General explicit difference formulas for numerical differentiation}.
\bjtitle{Journal of Computational and Applied Mathematics}
\bvolume{183}(\bissue{1}),
\bfpage{29}--\blpage{52}
(\byear{2005})
\doiurl{10.1016/j.cam.2004.12.026}
\end{barticle}
\endbibitem

\bibitem[\protect\citeauthoryear{Mu\~{n}oz{-}Coreas and Thapliyal}{2022}]{MunosCoreas2022}
\begin{bbook}
\bauthor{\bsnm{Mu\~{n}oz{-}Coreas}, \binits{E.}},
\bauthor{\bsnm{Thapliyal}, \binits{H.}}:
\bbtitle{{Everything You Always Wanted to Know about Quantum Circuits}},
pp. \bfpage{1}--\blpage{17}.
\bpublisher{John Wiley \& Sons, Ltd},
\blocation{Hoboken, NJ}
(\byear{2022}).
\doiurl{10.1002/047134608X.W8440}
\end{bbook}
\endbibitem

\bibitem[\protect\citeauthoryear{H{\"a}ner et~al.}{2018}]{haner2018optimizing}
\begin{barticle}
\bauthor{\bsnm{H{\"a}ner}, \binits{T.}},
\bauthor{\bsnm{Roetteler}, \binits{M.}},
\bauthor{\bsnm{Svore}, \binits{K.M.}}:
\batitle{Optimizing quantum circuits for arithmetic}.
\bjtitle{arXiv preprint arXiv:1805.12445}
(\byear{2018})
\doiurl{10.48550/arXiv.1805.12445}
\end{barticle}
\endbibitem

\bibitem[\protect\citeauthoryear{Grover and Rudolph}{2002}]{grover2002creating}
\begin{barticle}
\bauthor{\bsnm{Grover}, \binits{L.}},
\bauthor{\bsnm{Rudolph}, \binits{T.}}:
\batitle{Creating superpositions that correspond to efficiently integrable probability distributions}.
\bjtitle{arXiv preprint quant-ph/0208112}
(\byear{2002})
\doiurl{10.48550/arXiv.quant-ph/0208112}
\end{barticle}
\endbibitem

\bibitem[\protect\citeauthoryear{Sanders et~al.}{2019}]{Sanders2019}
\begin{barticle}
\bauthor{\bsnm{Sanders}, \binits{Y.R.}},
\bauthor{\bsnm{Low}, \binits{G.H.}},
\bauthor{\bsnm{Scherer}, \binits{A.}},
\bauthor{\bsnm{Berry}, \binits{D.W.}}:
\batitle{{Black-Box Quantum State Preparation without Arithmetic}}.
\bjtitle{Phys. Rev. Lett.}
\bvolume{122},
\bfpage{020502}
(\byear{2019})
\doiurl{10.1103/PhysRevLett.122.020502}
\end{barticle}
\endbibitem

\bibitem[\protect\citeauthoryear{Zoufal et~al.}{2019}]{zoufal2019quantum}
\begin{barticle}
\bauthor{\bsnm{Zoufal}, \binits{C.}},
\bauthor{\bsnm{Lucchi}, \binits{A.}},
\bauthor{\bsnm{Woerner}, \binits{S.}}:
\batitle{Quantum generative adversarial networks for learning and loading random distributions}.
\bjtitle{npj Quantum Information}
\bvolume{5}(\bissue{1}),
\bfpage{103}
(\byear{2019})
\doiurl{10.48550/arXiv.2205.00519}
\end{barticle}
\endbibitem

\bibitem[\protect\citeauthoryear{Holmes and Matsuura}{2020}]{Holmes2020}
\begin{bchapter}
\bauthor{\bsnm{Holmes}, \binits{A.}},
\bauthor{\bsnm{Matsuura}, \binits{A.Y.}}:
\bctitle{Efficient quantum circuits for accurate state preparation of smooth, differentiable functions}.
In: \bbtitle{2020 IEEE International Conference on Quantum Computing and Engineering (QCE)},
pp. \bfpage{169}--\blpage{179}
(\byear{2020}).
\doiurl{10.1109/QCE49297.2020.00030}
\end{bchapter}
\endbibitem

\bibitem[\protect\citeauthoryear{Rattew and Koczor}{2022}]{rattew2022preparing}
\begin{barticle}
\bauthor{\bsnm{Rattew}, \binits{A.G.}},
\bauthor{\bsnm{Koczor}, \binits{B.}}:
\batitle{Preparing arbitrary continuous functions in quantum registers with logarithmic complexity}.
\bjtitle{arXiv preprint arXiv:2205.00519}
(\byear{2022})
\doiurl{10.48550/arXiv.2205.00519}
\end{barticle}
\endbibitem

\bibitem[\protect\citeauthoryear{McArdle et~al.}{2022}]{mcardle2022quantum}
\begin{barticle}
\bauthor{\bsnm{McArdle}, \binits{S.}},
\bauthor{\bsnm{Gily{\'e}n}, \binits{A.}},
\bauthor{\bsnm{Berta}, \binits{M.}}:
\batitle{Quantum state preparation without coherent arithmetic}.
\bjtitle{arXiv preprint arXiv:2210.14892}
(\byear{2022})
\doiurl{10.48550/arXiv.2210.14892}
\end{barticle}
\endbibitem

\bibitem[\protect\citeauthoryear{Marin-Sanchez et~al.}{2023}]{MarinSanchez2023}
\begin{barticle}
\bauthor{\bsnm{Marin-Sanchez}, \binits{G.}},
\bauthor{\bsnm{Gonzalez-Conde}, \binits{J.}},
\bauthor{\bsnm{Sanz}, \binits{M.}}:
\batitle{Quantum algorithms for approximate function loading}.
\bjtitle{Phys. Rev. Res.}
\bvolume{5},
\bfpage{033114}
(\byear{2023})
\doiurl{10.1103/PhysRevResearch.5.033114}
\end{barticle}
\endbibitem

\bibitem[\protect\citeauthoryear{Moosa et~al.}{2023}]{Moosa_2024}
\begin{barticle}
\bauthor{\bsnm{Moosa}, \binits{M.}},
\bauthor{\bsnm{Watts}, \binits{T.W.}},
\bauthor{\bsnm{Chen}, \binits{Y.}},
\bauthor{\bsnm{Sarma}, \binits{A.}},
\bauthor{\bsnm{McMahon}, \binits{P.L.}}:
\batitle{Linear-depth quantum circuits for loading fourier approximations of arbitrary functions}.
\bjtitle{Quantum Science and Technology}
\bvolume{9}(\bissue{1}),
\bfpage{015002}
(\byear{2023})
\doiurl{10.1088/2058-9565/acfc62}
\end{barticle}
\endbibitem

\bibitem[\protect\citeauthoryear{Brassard et~al.}{2002}]{brassard2002}
\begin{barticle}
\bauthor{\bsnm{Brassard}, \binits{G.}},
\bauthor{\bsnm{Hoyer}, \binits{P.}},
\bauthor{\bsnm{Mosca}, \binits{M.}},
\bauthor{\bsnm{Tapp}, \binits{A.}}:
\batitle{Quantum amplitude amplification and estimation}.
\bjtitle{Contemporary Mathematics}
\bvolume{305},
\bfpage{53}
(\byear{2002})
\doiurl{10.1090/conm/305/05215}
\end{barticle}
\endbibitem

\bibitem[\protect\citeauthoryear{Suzuki et~al.}{2020}]{suzuki2020amplitude}
\begin{barticle}
\bauthor{\bsnm{Suzuki}, \binits{Y.}},
\bauthor{\bsnm{Uno}, \binits{S.}},
\bauthor{\bsnm{Raymond}, \binits{R.}},
\bauthor{\bsnm{Tanaka}, \binits{T.}},
\bauthor{\bsnm{Onodera}, \binits{T.}},
\bauthor{\bsnm{Yamamoto}, \binits{N.}}:
\batitle{Amplitude estimation without phase estimation}.
\bjtitle{Quantum Inf. Process.}
\bvolume{19},
\bfpage{75}
(\byear{2020})
\doiurl{10.1007/s11128-019-2565-2}
\end{barticle}
\endbibitem

\bibitem[\protect\citeauthoryear{Jerrum et~al.}{1986}]{jerrum1986}
\begin{barticle}
\bauthor{\bsnm{Jerrum}, \binits{M.R.}},
\bauthor{\bsnm{Valiant}, \binits{L.G.}},
\bauthor{\bsnm{Vazirani}, \binits{V.V.}}:
\batitle{Random generation of combinatorial structures from a uniform distribution}.
\bjtitle{Theoretical Computer Science}
\bvolume{43},
\bfpage{169}
(\byear{1986})
\end{barticle}
\endbibitem

\bibitem[\protect\citeauthoryear{Miyamoto}{2025}]{Miyamoto2025Dividing}
\begin{barticle}
\bauthor{\bsnm{Miyamoto}, \binits{K.}}:
\batitle{Dividing quantum circuits for the time evolution of stochastic processes by orthogonal-series-density estimation}.
\bjtitle{Phys. Rev. A}
\bvolume{111},
\bfpage{042431}
(\byear{2025})
\doiurl{10.1103/PhysRevA.111.042431}
\end{barticle}
\endbibitem

\bibitem[\protect\citeauthoryear{Shimizu and Yamada}{1972}]{ShimizuYamada1972}
\begin{barticle}
\bauthor{\bsnm{Shimizu}, \binits{H.}},
\bauthor{\bsnm{Yamada}, \binits{T.}}:
\batitle{Phenomenological equations of motion of muscular contraction}.
\bjtitle{Progress of Theoretical Physics}
\bvolume{47}(\bissue{1}),
\bfpage{350}--\blpage{351}
(\byear{1972})
\doiurl{10.1143/PTP.47.350}
\end{barticle}
\endbibitem

\bibitem[\protect\citeauthoryear{Rackauckas and Nie}{2017}]{rackauckas2017differentialequations}
\begin{botherref}
\oauthor{\bsnm{Rackauckas}, \binits{C.}},
\oauthor{\bsnm{Nie}, \binits{Q.}}:
Differential{E}quations.jl--a performant and feature-rich ecosystem for solving differential equations in {J}ulia.
Journal of Open Research Software
\textbf{5}(1)
(2017)
\end{botherref}
\endbibitem

\bibitem[\protect\citeauthoryear{Shinomoto and Kuramoto}{1986}]{Shinomoto1986}
\begin{barticle}
\bauthor{\bsnm{Shinomoto}, \binits{S.}},
\bauthor{\bsnm{Kuramoto}, \binits{Y.}}:
\batitle{Phase transitions in active rotator systems}.
\bjtitle{Progress of Theoretical Physics}
\bvolume{75}(\bissue{5}),
\bfpage{1105}--\blpage{1110}
(\byear{1986})
\doiurl{10.1143/PTP.75.1105}
\end{barticle}
\endbibitem

\bibitem[\protect\citeauthoryear{Sakaguchi et~al.}{1988}]{Sakaguchi1988}
\begin{barticle}
\bauthor{\bsnm{Sakaguchi}, \binits{H.}},
\bauthor{\bsnm{Shinomoto}, \binits{S.}},
\bauthor{\bsnm{Kuramoto}, \binits{Y.}}:
\batitle{Phase transitions and their bifurcation analysis in a large population of active rotators with mean-field coupling}.
\bjtitle{Progress of Theoretical Physics}
\bvolume{79}(\bissue{3}),
\bfpage{600}--\blpage{607}
(\byear{1988})
\doiurl{10.1143/PTP.79.600}
\end{barticle}
\endbibitem

\end{thebibliography}

\end{document}